\documentclass[11pt,draftcls,onecolumn]{IEEEtran}

\usepackage{epsf}
\usepackage{graphicx}
\usepackage{amsfonts,amssymb,amsmath}
\usepackage{algorithmic,algorithm}
\usepackage{subfigure}
\usepackage{latexsym}
\usepackage{bm}
\usepackage[usenames]{color}
\usepackage{url}
\usepackage{enumerate}

\usepackage[normalem]{ulem}

\newcommand{\expect}[1]{\mathbb{E}\left[#1\right]}
\newtheorem{theorem}{Theorem}
\newtheorem{lemma}{Lemma}

\newtheorem{assumption}{Assumption}
\newtheorem{corollary}{Corollary}

\begin{document}

\title{Recoverability of Group Sparse Signals from Corrupted Measurements via Robust Group Lasso}
\author{\authorblockN{Xiaohan Wei, Qing Ling, and Zhu Han}
\thanks{Xiaohan Wei is with Department of Electrical Engineering,
University of Southern California. Email: xiaohanw@usc.edu. Qing
Ling is With Department of Automation, University of Science and
Technology of China. Email: qingling@mail.ustc.edu.cn. Zhu Han is
with Department of Electrical and Computer Engineering, University
of Houston. Email: zhan2@uh.edu.}}

\maketitle

\begin{abstract}
This paper considers the problem of recovering a group sparse
signal matrix $\mathbf{Y} = [\mathbf{y}_1, \cdots, \mathbf{y}_L]$
from sparsely corrupted measurements $\mathbf{M} =
[\mathbf{A}_{(1)}\mathbf{y}_{1}, \cdots,
\mathbf{A}_{(L)}\mathbf{y}_{L}] + \mathbf{S}$, where
$\mathbf{A}_{(i)}$'s are known sensing matrices and $\mathbf{S}$
is an unknown sparse error matrix. A robust group lasso (RGL)
model is proposed to recover $\mathbf{Y}$ and $\mathbf{S}$ through
simultaneously minimizing the $\ell_{2,1}$-norm of $\mathbf{Y}$
and the $\ell_1$-norm of $\mathbf{S}$ under the measurement
constraints. We prove that $\mathbf{Y}$ and $\mathbf{S}$ can be
exactly recovered from the RGL model with a high probability for a
very general class of $\mathbf{A}_{(i)}$'s.
\end{abstract}

\section{Introduction}
\label{sec:intro}

Consider the problem of recovering a group sparse signal matrix
$\mathbf{Y} = [\mathbf{y}_1, \cdots, \mathbf{y}_L] \in
\mathcal{R}^{n \times L}$ from sparsely corrupted measurements
\begin{equation}\label{ex1}
    \mathbf{M} = [\mathbf{A}_{(1)}\mathbf{y}_{1}, \cdots, \mathbf{A}_{(L)}\mathbf{y}_{L}] +
    \mathbf{S},
\end{equation}
where $\mathbf{M} = [\mathbf{m}_1, \cdots, \mathbf{m}_L] \in
\mathcal{R}^{m \times L}$ is a measurement matrix,
$\mathbf{A}_{(i)} \in \mathcal{R}^{m \times n}$ is the $i$-th
sensing matrix, and $\mathbf{S} = [\mathbf{s}_1, \cdots,
\mathbf{s}_L] \in \mathcal{R}^{m \times L}$ is an unknown sparse
error matrix. The error matrix $\mathbf{S}$ is sparse as it has
only a small number of nonzero entries. The signal matrix
$\mathbf{Y}$ is group sparse, meaning that $\mathbf{Y}$ is sparse
and its nonzero entries appear in a small number of common rows.

Given $\mathbf{M}$ and $\mathbf{A}_{(i)}$'s, our goal is to recover
$\mathbf{Y}$ and $\mathbf{S}$ from the linear measurement equation
\eqref{ex1}. In this paper, we propose to accomplish the recovery
task through solving the following robust group lasso (RGL) model
\begin{align}\label{ee3}
  \min\limits_{\mathbf{Y}, \mathbf{S}} & \quad \|\mathbf{Y}\|_{2,1} + \lambda \|\mathbf{S}\|_{1}, \nonumber \\
  s.t.                                 & \quad \mathbf{M} = [\mathbf{A}_{(1)}\mathbf{y}_{1},
  \cdots, \mathbf{A}_{(L)}\mathbf{y}_{L}] + \mathbf{S}.
\end{align}
Denoting $y_{ij}$ and $s_{ij}$ as the $(i,j)$-th entries of
$\mathbf{Y}$ and $\mathbf{S}$, respectively, $\|\mathbf{Y}\|_{2,1}
\triangleq \sum_{i=1}^n \sqrt{\sum_{j=1}^L y_{ij}^2}$ is defined
as the $\ell_{2,1}$-norm of $\mathbf{Y}$ and $\|\mathbf{S}\|_1
\triangleq \sum_{i=1}^m \sum_{j=1}^L |s_{ij}|$ is defined as the
$\ell_1$-norm of $\mathbf{S}$. Minimizing the $\ell_{2,1}$-norm
term promotes group sparsity of $\mathbf{Y}$ while minimizing the
$\ell_1$-norm term promotes sparsity of $\mathbf{S}$; $\lambda$ is
a nonnegative parameter to balance the two terms. We prove that
solving the RGL model in \eqref{ee3}, which is a convex program,
enables exact recovery of $\mathbf{Y}$ and $\mathbf{S}$ with high
probability, given that $\mathbf{A}_{(i)}$'s satisfy certain
conditions.

\subsection{From Group Lasso to Robust Group Lasso}
\label{sec:intro1}

Sparse signal recovery has attracted research interests in the
signal processing and optimization communities during the past few
years. Various sparsity models have been proposed to better
exploit the sparse structures of high-dimensional data, such as
sparsity of a vector \cite{donoho2006}, \cite{candes2006}, group
sparsity of vectors \cite{yuan2007}, and low-rankness of a matrix
\cite{candes2008}. For more topics related to sparse signal
recovery, readers are referred to the recent survey paper
\cite{review_2014}.

In this paper we are interested in the recovery of group sparse
(also known as block sparse \cite{eldar2010} or jointly sparse
\cite{eldar2012}) signals which finds a variety of applications
such as direction-of-arrival estimation \cite{malioutov2005},
\cite{wei2012}, collaborative spectrum sensing \cite{tian2011,
meng2011, giannakis2011} and motion detection
\cite{motion_detection}. A signal matrix $\mathbf{Y} =
[\mathbf{y}_1, \cdots, \mathbf{y}_L] \in \mathcal{R}^{n\times L}$
is called $k$-group sparse if $k$ rows of $\mathbf{Y}$ are
nonzero. A measurement matrix $\mathbf{M} = [\mathbf{m}_1, \cdots,
\mathbf{m}_L] \in\mathcal{R}^{m\times L}$ is taken from linear
projections $\mathbf{m}_{i}=\mathbf{A}_{(i)}\mathbf{y}_i$,
$i=1,\cdots,L$, where $\mathbf{A}_{(i)}\in\mathcal{R}^{m\times n}$
is a sensing matrix. In order to recover $\mathbf{Y}$ from
$\mathbf{A}_{(i)}$'s and $\mathbf{M}$, the standard
$\ell_{2,1}$-norm minimization formulation proposes to solve a
convex program
\begin{align}\label{ee1}
  &\min_{\mathbf{Y}}~\|\mathbf{Y}\|_{2,1}, \nonumber\\
  &s.t.             ~~~\mathbf{M}=[\mathbf{A}_{(1)}\mathbf{y}_{1}, \cdots,
  \mathbf{A}_{(L)}\mathbf{y}_{L}].
\end{align}
This is a straightforward extension from the canonical
$\ell_1$-norm minimization formulation that recovers a sparse
vector. Theoretical guarantee of exact recovery has been developed
based on the restricted isometric property (RIP) of
$\mathbf{A}_{(i)}$'s \cite{eldar2009}, and a reduction of the
required number of measurements can also be achieved through
simultaneously minimizing the $\ell_{2,1}$-norm and the nuclear
norm of $\mathbf{Y}$; see \cite{Vandergheynst2012} and
\cite{simultaneous_2}.

Consider that in practice the measurements are often corrupted by
random noise, resulting in
$\mathbf{M}=[\mathbf{A}_{(1)}\mathbf{y}_1, \cdots,
\mathbf{A}_{(L)}\mathbf{y}_L]+\mathbf{N}$ where $\mathbf{N} =
[\mathbf{n}_1, \cdots, \mathbf{n}_L] \in\mathcal{R}^{m\times L}$
is a noise matrix. To address the noise-corrupted case, the group
lasso model in \cite{yuan2007} solves
\begin{align}\label{ee2gl}
  \min\limits_{\mathbf{Y}, \mathbf{E}} & \quad \|\mathbf{Y}\|_{2,1} + \gamma \|\mathbf{N}\|_F^2, \nonumber \\
  s.t.                                 & \quad \mathbf{M} = [\mathbf{A}_{(1)}\mathbf{y}_{1},
  \cdots, \mathbf{A}_{(L)}\mathbf{y}_{L}] + \mathbf{N},
\end{align}
where $\gamma$ is a nonnegative parameter and $\|\mathbf{N}\|_F$
is the Frobenius norm of $\mathbf{N}$. An alternative to
(\ref{ee2gl}) is
\begin{align}\label{ee2}
  &\min_{\mathbf{Y}}~\|\mathbf{Y}\|_{2,1}, \nonumber\\
  &s.t.             ~~~\|\mathbf{M}-[\mathbf{A}_{(1)}\mathbf{y}_{1},
  \cdots, \mathbf{A}_{(L)}\mathbf{y}_{L}]\|_F^2 \leq\varepsilon^2,
\end{align}
where $\varepsilon$ controls the noise level. It has been shown in
\cite{eldar2009} that if the sensing matrices $\mathbf{A}_{(i)}$'s
satisfy RIP, then the distance between the solution to (\ref{ee2})
and the true signal matrix, which is measured by the Frobenius
norm, is within a constant multiple of $\varepsilon$.

The exact recovery guarantee for (\ref{ee2}) is elegant, but works
only if the noise level $\varepsilon$ is sufficiently small. However,
in many practical applications, some of the measurements may be
seriously contaminated or even missing due to uncertainties such
as sensor failures and transmission errors. Meanwhile, this kind
of measurement errors are often sparse (see
\cite{giannakis2011(2)} for detailed discussions). In this case,
the exact recovery guarantee does not hold and the solution of
(\ref{ee2}) can be far away from the true signal matrix.

The need of handling large but sparse measurement errors in the
group sparse signal recovery problem motivates the RGL model
(\ref{ee3}), which has found successful applications in, for
example, the cognitive network sensing problem
\cite{giannakis2011(2)}. In (\ref{ee3}), the measurement matrix
$\mathbf{M}$ is contaminated by a sparse error matrix $\mathbf{S}
= [\mathbf{s}_1, \cdots, \mathbf{s}_L] \in\mathcal{R}^{m\times L}$
whose nonzero entries might be unbounded. Through simultaneously
minimizing the $\ell_{2,1}$-norm of $\mathbf{Y}$ and the $\ell_1$
norm of $\mathbf{S}$, we expect to recover the group sparse signal
matrix $\mathbf{Y}$ and the sparse error matrix $\mathbf{S}$.

The RGL model \eqref{ee3} is tightly related to robust lasso and
robust principle component analysis (RPCA), both of which have
been proved effectively in recovering true signal from sparse
gross corruptions. The robust lasso model, which has been
discussed extensively in \cite{wright2009}, \cite{li2012},
\cite{tran2010}, minimizes the $\ell_1$-norm of a sparse signal vector
and the $\ell_1$-norm of a sparse error vector simultaneously in
order to remove sparse corruptions. Whereas the RPCA model,
which is first proposed in \cite{candes2009} and then extended by
\cite{Ganesh2012} and \cite{caramanis2013},  recovers a low rank
matrix by minimizing the nuclear norm of signal matrix plus the
$\ell_1$-norm of sparse error matrix.


\subsection{Contribution and Paper Organization}
\label{sec:intro2}

This paper proposes the RGL model for recovering the group sparse
signal from unbounded sparse corruptions and
proves that with a high probability, the proposed RGL
model \eqref{ee3} exactly recovers the group sparse signal matrix
and the sparse error matrix simultaneously under certain
restrictions on the measurement matrix for a very general
class of sample matrices.

The rest of this paper is organized as follows. Section
\ref{sec:main} provides the main result (see Theorem
\ref{recovery_guarantee}) on the recoverability of the RGL model
\eqref{ee3} under the assumptions on the sensing matrices and the true
signal and error matrices (see Assumptions
\ref{definition_sensing_matrix}-\ref{assumption_true}). Section
\ref{sec:main} also introduces several supporting lemmas and
corollaries (See Lemmas \ref{lemma_1}-\ref{lemma_BEQ_2} and
Corollaries \ref{coro_1}-\ref{BEQ_corollary_plusplus}). Section
\ref{sec:certificate} gives the dual certificates of \eqref{ee3},
one is exact (see Theorem \ref{theorem_exact_duality}) and the other
is inexact (see Theorem \ref{theorem_inexact_duality}), which are
sufficient conditions guaranteing exact recovery from the RGL
model with a high probability. Their proofs are based on two
supporting lemmas (see Lemmas \ref{equivalence_relation}-
\ref{null_space_lemma}). Section \ref{sec:golfing} proves that the
inexact dual certificate of \eqref{ee3}
can be satisfied through a constructive manner (see Theorem
\ref{theorem_dual_certificate} and Lemma
\ref{batch_shrink_lemma}). This way, we prove the main result
given in Section \ref{sec:main}. Section \ref{sec:conclusion}
concludes the paper.

\subsection{Notations}
\label{sec:intro3}

We introduce several notations that are used in the subsequent
sections. Bold uppercase letters denote matrices, whereas bold
lowercase letters with subscripts and superscripts stand for
column vectors and row vectors, respectively. For a matrix
$\mathbf{U}$, we denote $\mathbf{u}_i$ as its $i$-th column,
$\mathbf{u}^i$ as its $j$-th row, and $u_{ij}$ as its $(i,j)$-th
element. For a given vector $\mathbf{u}$, we denote $u_i$ as its
$i$-th element. The notations $\{\mathbf{U}_{(i)}\}$ and
$\{\mathbf{u}_{(i)}\}$ denote the family of matrices and vectors
indexed by $i$, respectively. The notations
$\{\mathbf{U}_{(i,j)}\}$ and $\{\mathbf{u}_{(i,j)}\}$ denote the
family of matrices and vectors indexed by $(i,j)$, respectively.
$\textrm{vec}(\cdot)$ is the vectorizing operator that stacks the
columns of a matrix one after another. $\{\cdot\}^{'}$ denotes the
transpose operator. $\mathbf{diag}\{\cdot\}$ represents a diagonal
matrix and $\mathbf{BLKdiag}\{\cdot\}$ represents a block diagonal
matrix. The notation $\langle\cdot,\cdot\rangle$ denotes the inner
product, when applying to two matrices $\textbf{U}$ and
$\textbf{V}$. $\textrm{sgn}(\mathbf{u})$ and
$\textrm{sgn}{(\mathbf{U})}$ are sign vector and sign matrix for
$\mathbf{u}$ and $\mathbf{U}$, respectively.

Additionally, we use several standard matrix and vector norms. For
a vector $\mathbf{u}\in\mathcal{R}^n$, define
\begin{itemize}
  \item $\ell_2$-norm: $\|\mathbf{u}\|_2=\sqrt{\sum_{j=1}^n
  u_{j}^2}$.
  \item $\ell_1$-norm: $\|\mathbf{u}\|_1=\sum_{j=1}^n |u_{j}|$.
\end{itemize}
For a matrix $\mathbf{U}\in\mathcal{R}^{m\times n}$, define
\begin{itemize}
  \item $\ell_{2,1}$-norm:
  $\|\mathbf{U}\|_{2,1}=\sum_{i=1}^{m}\sqrt{\sum_{j=1}^nu_{ij}^2}$.
  \item$\ell_{2,\infty}$-norm:
      $\|\mathbf{U}\|_{2,\infty}=\max_{i}\sqrt{\sum_{j=1}^nu_{ij}^2}$.
  \item $\ell_{1}$-norm:
  $\|\mathbf{U}\|_1=\sum_{i=1}^m \sum_{j=1}^n|u_{ij}|$.
  \item Frobenius norm:
  $\|\mathbf{U}\|_F=\sqrt{\sum_{i=1}^m \sum_{j=1}^n u_{ij}^2}$.
  \item $\ell_\infty$-norm:
  $\|\mathbf{U}\|_{\infty}=\max_{i,j}|u_{ij}|$.
\end{itemize}
Also, we use the notation $\|\mathbf{U}\|_{(p,q)}$ to denote the induced norms, which stands for
\[\|\mathbf{U}\|_{(p,q)}=\max_{\mathbf{x}\in\mathcal{R}^{n}}\frac{\|\mathbf{Ux}\|_{p}}{\|\mathbf{x}\|_{q}}.\]

For the signal matrix $\mathbf{Y}\in\mathcal{R}^{n\times L}$ and
noise matrix $\mathbf{S}\in\mathcal{R}^{m\times L}$, we use the
following set notations throughout the paper.
\begin{itemize}
\item $T$: The row group support (namely, the set of row
coordinates corresponding to the nonzero rows of the signal
matrix) whose cardinality is denoted as $k_T = |T|$. \item $T^c$:
The complement of $T$ (namely, $\{1,\cdots,n\}\setminus T$). \item
$\Omega$: The support of error matrix (namely, the set of
coordinates corresponding to the nonzero elements of the error
matrix) whose cardinality is denoted as $k_\Omega =|\Omega|$.
\item $\Omega^c$: The complement of $\Omega$ (namely,
$\{1,\cdots,n\}\times\{1,\cdots,L\}\setminus\Omega$). \item
$\Omega_i$: The support of the $i$-th column the error matrix
whose cardinality is denoted as $k_{\Omega_i} =|\Omega_i|$. \item
$\Omega_i^c$: The complement of $\Omega_i$ (namely,
$\{1,\cdots,n\}\setminus\Omega_i$). \item $\Omega^*_{i}$: An
arbitrary fixed subset of $\Omega^c_i$ with cardinality
$m-k_{\max}$, where $k_{\max}=\max_i k_{\Omega_i}$. Intuitively,
$\Omega^*_{i}$ stands for the \emph{maximal non-corrupted set}
across different $i\in\{1,\cdots, L\}$.
\end{itemize}
For any given matrices $\mathbf{U}\in\mathcal{R}^{m\times L}$,
$\mathbf{V}\in\mathcal{R}^{n\times L}$ and given vectors
$\mathbf{u}\in \mathcal{R}^m$, $\mathbf{v}\in \mathcal{R}^n$,
define the orthogonal projection operators as follows.
\begin{itemize}
\item $\mathcal{P}_{\Omega}\mathbf{U}$: The orthogonal projection
of matrix $\mathbf{U}$ onto $\Omega$ (namely, set every entry of
$\mathbf{U}$ whose coordinate belongs to $\Omega^c$ as 0 while
keep other entries unchanged). \item
$\mathcal{P}_{\Omega_i}\mathbf{u}$,
$\mathcal{P}_{\Omega_i^c}\mathbf{u}$,
$\mathcal{P}_{\Omega_{i}^*}\mathbf{u}$: The orthogonal projections
of $\mathbf{u}$ onto $\Omega_i$, $\Omega_i^c$, and $\Omega_{i}^*$,
respectively. \item $\mathcal{P}_T\mathbf{v}$: The orthogonal
projection of $\mathbf{v}$ onto $T$. \item
$\mathcal{P}_{\Omega_i}\mathbf{U}$,
$\mathcal{P}_{\Omega_i^c}\mathbf{U}$, and
$\mathcal{P}_{\Omega_{i}^*}\mathbf{U}$: The orthogonal projections
of each column of $\mathbf{U}$ onto $\Omega_i$, $\Omega_i^c$, and
$\Omega_{i}^*$, respectively (namely,
$\mathcal{P}_{\Omega_i}\mathbf{U}=[\mathcal{P}_{\Omega_i}\mathbf{u}_1,
\cdots, \mathcal{P}_{\Omega_i}\mathbf{u}_L]$,
    $\mathcal{P}_{\Omega_i^c}\mathbf{U}=[\mathcal{P}_{\Omega_i^c}\mathbf{u}_1, \cdots, \mathcal{P}_{\Omega_i^c}\mathbf{u}_L]$ and
    $\mathcal{P}_{\Omega_i^*}\mathbf{U}=[\mathcal{P}_{\Omega_i^*}\mathbf{u}_1, \cdots, \mathcal{P}_{\Omega_i^*}\mathbf{u}_L]$).
\item $\mathcal{P}_T\mathbf{V}$: The orthogonal projection of each
column of $\mathbf{V}$ onto $T$.
\end{itemize}
Furthermore, we admit a notational convention that for any
projection operator $\mathcal{P}$ and corresponding matrix
$\mathbf{U}$ (or vector $\mathbf{u}$), it holds
$$\mathbf{U}'\mathcal{P}=\left(\mathcal{P}\mathbf{U}\right)' (\textrm{or}~\mathbf{u}'\mathcal{P}=\left(\mathcal{P}\mathbf{u}\right)').$$
Finally, by saying an event occurs with a high probability, we mean
that the occurring probability of the event is at least
$1-Cn^{-1}$ where $C$ is a constant.

\section{Main Result of Exact Recovery}
\label{sec:main}

This section provides the theoretical performance guarantee of the
RGL model (\ref{ee3}). Section \ref{sec:main1} makes several
assumptions under which (\ref{ee3}) recovers the true group sparse
signal and sparse error matrices with a high probability. The main
result is summarized in Theorem \ref{recovery_guarantee}. Section
\ref{sec:main2} interprets the meanings of Theorem
\ref{recovery_guarantee} and explains its relations to previous
works. Section \ref{sec:main1} gives several measure concentration
inequalities that are useful in the proof of the main result.

\subsection{Assumptions and Main Result}
\label{sec:main1}

We start from several assumptions on the sensing matrices, as well
as the true group sparse signal and sparse error matrices.
Consider $L$ distributions $\{\mathcal{F}_{i}\}_{i=1}^L$ in
$\mathcal{R}^{n}$ and an independently sampled vector
$\mathbf{a}_{(i)}$ from each $\mathcal{F}_{i}$. The correlation
matrix is defined as
\begin{equation}
    \mathbf{\Sigma}_{(i)}=\expect{\mathbf{a}_{(i)}\mathbf{a}'_{(i)}}, \nonumber
\end{equation}
and the corresponding condition number is
\begin{equation}
    \kappa_{i}=\sqrt{\frac{\lambda_{\max}\{\mathbf{\Sigma}_{(i)}\}}{\lambda_{\min}\{\mathbf{\Sigma}_{(i)}\}}}, \nonumber
\end{equation}
where $\lambda_{\max}\{\cdot\}$ and $\lambda_{\min}\{\cdot\}$
denotes the largest and smallest eigenvalues of a matrix,
respectively. We use $\kappa_{\max} = \max_i \kappa_i$ to
represent the maximum condition number regarding a set of
covariance matrices. Observe that this condition number is finite
if and only if the covariance matrix is invertible, and is larger
than or equal to 1 in any case.

\begin{assumption}\label{definition_sensing_matrix}
For $i=1, \cdots, L$, define the $i$-th sensing matrix as
\begin{equation}
    \mathbf{A}_{(i)}\triangleq\frac{1}{\sqrt{m}}
    \left(
      \begin{array}{c}
        \mathbf{a}_{(i)1}' \\
        \vdots \\
        \mathbf{a}_{(i)m}' \\
      \end{array}
    \right) \in \mathcal{R}^{m \times n}. \nonumber
\end{equation}
Therein, $\{\mathbf{a}_{(i)1}, \cdots,\mathbf{a}_{(i)m}\}$ is
assumed to be a sequence of i.i.d. random vectors drawn from the
distribution $\mathcal{F}_{i}$ in $\mathcal{R}^{n}$.
\end{assumption}

By Assumption \ref{definition_sensing_matrix}, we suppose that
every sensing matrix $\mathbf{A}_{(i)}$ is randomly sampled from a
corresponding distribution $\mathcal{F}_{i}$. We proceed to assume
the properties of the distributions $\{\mathcal{F}_{i}\}_{i=1}^L$.

\begin{assumption}\label{assumption_population}
For each $i=1, \cdots, L$, the distribution $\mathcal{F}_{i}$ satisfies
the following two properties.
\begin{itemize}
  \item \textbf{Completeness}: The correlation matrix $\mathbf{\Sigma}_{(i)}$ is invertible.
  \item \textbf{Incoherence:}
Each sensing vector $\mathbf{a}_{(i)}$ sampled from
$\mathcal{F}_i$ satisfies
\begin{align}
&\max_{j\in\{1,\cdots,n\}}|\langle\mathbf{a}_{(i)},\mathbf{e}_k\rangle|\leq \sqrt{\mu_{i}}, \label{incoherence1}\\
&\max_{j\in\{1,\cdots,n\}}|\langle\mathbf{\Sigma}_{(i)}^{-1}\mathbf{a}_{(i)},\mathbf{e}_k\rangle|\leq
\sqrt{\mu_{i}}, \label{incoherence2}
\end{align}
for some fixed constant $\mu_i\geq1$, where $\{\mathbf{e}_k\}_{k=1}^n$ is the standard basis in $\mathcal{R}^{n}$.
\end{itemize}
\end{assumption}

We call $\mu_{i}$ as the incoherence parameter and use $\mu_{\max}
= \max_i \mu_i$ to denote the maximum incoherence parameter among
a set of $L$ distributions $\{\mathcal{F}_i\}_{i=1}^L$. Note that
this incoherence condition is stronger than the one originally
presented in \cite{candes2010}, which does not require
\eqref{incoherence2}. If one wants to get rid of
\eqref{incoherence2}, then some other restrictions must be imposed
on the sensing matrices (see \cite{gross2012} for related
results).

Observe that the bounds \eqref{incoherence1} and
\eqref{incoherence2} in Assumption \ref{assumption_population} are
meaningless unless we fix the scale of $\mathbf{a}_{(i)}$. Thus,
we have the following assumption.

\begin{assumption}\label{assumption_scaling}
The correlation matrix $\mathbf{\Sigma}_{(i)}$ satisfies
\begin{equation}\label{fixed_scaling}
\lambda_{\max}\{\mathbf{\Sigma}_{(i)}\}=
\lambda_{\min}\{\mathbf{\Sigma}_{(i)}\}^{-1},
\end{equation}
for any $\mathcal{F}_i,~i=1,\cdots,L$.
\end{assumption}

Given any complete $\mathcal{F}_i$, \eqref{fixed_scaling} can
always be achieved by scaling $\mathbf{a}_{(i)}$ up or down. This
is true because if we scale $\mathbf{a}_{(i)}$ up, then
$\lambda_{\max}\{\mathbf{\Sigma}_{(i)}\}$ increases and
$\lambda_{\min}\{\mathbf{\Sigma}_{(i)}\}^{-1}$ decreases. Observe
that the optimization problem \eqref{ee3} is invariant under
scaling. Thus, Assumption \ref{assumption_scaling} does not pose
any extra constraint.

Additionally, we denote $\overline{\mathbf{Y}}$ and
$\overline{\mathbf{S}}$ as the true group sparse signal and sparse
error matrices to recover, respectively. The assumption on
$\overline{\mathbf{Y}}$ and $\overline{\mathbf{S}}$ is given as
below.

\begin{assumption}\label{assumption_true}
The true signal matrix $\overline{\mathbf{Y}}$ and error matrix
$\overline{\mathbf{S}}$ satisfy the following two properties.
\begin{itemize}
  \item The row group support of $\overline{\mathbf{Y}}$ and the support of $\overline{\mathbf{S}}$ are fixed and denoted as $T$ and $\Omega$, respectively.
  \item The signs of the elements of $\overline{\mathbf{Y}}$ and $\overline{\mathbf{S}}$ are i.i.d. and equally likely to be $+1$ or $-1$.
\end{itemize}
\end{assumption}

Under the assumptions stated above, we have the following main
theorem on the recoverability of the RGL model \eqref{ee3}.

\begin{theorem}\label{recovery_guarantee}
Under Assumptions
\ref{definition_sensing_matrix}-\ref{assumption_true}, the
solution pair ($\hat{\mathbf{Y}}$, $\hat{\mathbf{S}}$) to the
optimization problem (\ref{ee3}) is exact and unique with
probability at least $1-(16+2e^{\frac14})n^{-1}$, provided that
$\lambda=\frac{1}{\sqrt{\log n}}$, $k_TL\leq n$,
\begin{equation}\label{ee4}
    k_{T}\leq\alpha\frac{m}{\mu_{\max}\kappa_{\max}\log^{2} n}, \quad k_{\Omega}\leq\beta\frac{m}{\mu_{\max}}, \quad k_{\max}\leq\gamma
    \frac{m}{\kappa_{\max}}.
\end{equation}
Here $\mu_{\max} \triangleq \max_i\mu_{i}$, $\kappa_{\max}
\triangleq \max_i\kappa_{i}$, $k_{\max} \triangleq
\max_ik_{\Omega_i}$, and $\alpha\leq\frac{1}{9600}$,
$\beta\leq\frac{1}{3136}$, $\gamma\leq\frac{1}{4}$ are all
positive constants\footnote{The bounds on $\alpha$, $\beta$,
$\gamma$ are chosen such that all the requirements on these
constants in the subsequent lemmas and theorems are met.}.
\end{theorem}

\subsection{Interpretations of Theorem \ref{recovery_guarantee} and Relations to Previous Works}
\label{sec:main2}

Now we discuss what Theorem \ref{recovery_guarantee} implies.
First, it infers that when the signal matrix
$\overline{\mathbf{Y}}$ is sufficiently group sparse and the error
matrix $\overline{\mathbf{S}}$ sufficiently sparse (see the bounds
on $k_T$, $k_{\Omega}$, and $k_{\max}$), then with high
probability we are able to exactly recover them. Second, observe
that the group sparsity does not depend on $L$, the number of
columns of the signal matrix, as long as $L$ is not too large (see
the bound on $k_T$). This demonstrates the ability of the RGL
model in recovering group sparse signals even though each nonzero
row is not sparse. Last, to keep the proof simple, we do not
optimize the constants $\alpha$, $\beta$, and $\gamma$. However,
it is possible to increase the values of the constants and
consequently relax the requirements on the sparsity patterns.

Theorem \ref{recovery_guarantee} is a result of RIPless analysis,
which shares the same limitation as all other RIPless analyses. To
be specific, Theorem \ref{recovery_guarantee} only holds for
arbitrary but \emph{fixed} $\overline{\mathbf{Y}}$ and
$\overline{\mathbf{S}}$ (except that the elements of
$\overline{\mathbf{Y}}$ and $\overline{\mathbf{S}}$ have uniform
random signs by Assumption \ref{assumption_true}). If we expect to
have a uniform recovery guarantee here (namely, considering random
sensing matrices as well as signal and error matrices with random
supports), then certain stronger assumptions must be made on the
sensing matrices such as the RIP condition
\cite{Vandergheynst2012}.

The proof of Theorem \ref{recovery_guarantee} is based on the
construction of an inexact dual certificate through the golfing
scheme. The golfing scheme was first introduced in
\cite{gross2009} for low rank matrix recovery. Subsequently,
\cite{candes2010} and \cite{gross2012} refined and used the scheme
to prove the lasso recovery guarantee. The work \cite{li2012}
generalized it to mix-norm recovery. In this paper, we consider a
new mix-norm problem, namely, summation of the $\ell_{2,1}$-norm
and the $\ell_1$-norm.

\subsection{Measure Concentration Inequalities}
\label{sec:main3}

Below we give several measure concentration inequalities that are
useful in the proofs of the paper. We begin with two lemmas on
Berstein inequalities from \cite{candes2010}, whose proofs are
omitted for brevity. The first one is a matrix Berstein
inequality.

\begin{lemma} \label{lemma_1}
(Matrix Berstein Inequality) Consider a finite sequence of
independent random matrices
$\{\mathbf{M}_{(j)}\in\mathcal{R}^{d\times d}\}$. Assume that
every random matrix satisfies $\expect{\mathbf{M}_{(j)}}=0$ and
$\|\mathbf{M}_{(j)}\|_{(2,2)}\leq B$ almost surely. Define
\begin{equation}
    \sigma^2 \triangleq \max\left\{\left\|\sum_{j}\expect{\mathbf{M}_{(j)}'\mathbf{M}_{(j)}}\right\|_{(2,2)},
    ~\left\|\sum_{j}\expect{\mathbf{M}_{(j)}\mathbf{M}_{(j)}'}\right\|_{(2,2)}\right\}. \nonumber
\end{equation}
Then, for all $t\geq0$, we have
\begin{equation}
    Pr\left\{\left\|\sum_{j}\mathbf{M}_{(j)}\right\|_{(2,2)}\geq t\right\}\leq 2d\exp\left(-\frac{t^2/2}{\sigma^2+Bt/3}\right). \nonumber
\end{equation}
\end{lemma}

We also need a vector form of the Berstein inequality.

\begin{lemma} \label{lemma_2}
(Vector Berstein Inequality) Consider a finite sequence of
independent random vectors
$\{\mathbf{g}_{(j)}\in\mathcal{R}^{d}\}$. Assume that every random
vector satisfies $\expect{\mathbf{g}_{(j)}}=0$ and
$\|\mathbf{g}_{(j)}\|_2\leq B$ almost surely. Define
$\sigma^{2}\triangleq \sum_{k}\expect{\|\mathbf{g}_{(j)}\|_2^2}$.
Then, for all $0\leq t\leq \sigma^2/B$, we have
\begin{equation}
Pr\left(\left\|\sum_{j}\mathbf{g}_{(j)}\right\|_{2}\geq
t\right)\leq
\textrm{exp}\left(-\frac{t^2}{8\sigma^2}+\frac{1}{4}\right).
\nonumber
\end{equation}
\end{lemma}

Next, we use the matrix Berstein inequality to prove its extension
on a block anisotropic matrix.

\begin{lemma}\label{lemma_BEQ_1}
Consider a matrix $\mathbf{A}_{(i)}$ satisfying the model
described in Section \ref{sec:main1}, and denote
$\tilde{\mathbf{A}}_{(i)}=\mathbf{\Sigma}_{(i)}^{-1}\mathbf{A}_{(i)}'\mathcal{P}_{\Omega_{i}^*}\mathbf{A}_{(i)}$.
For any $\tau>0$, it holds
\begin{equation}
    Pr\left\{\left\|\mathcal{P}_{T}\left(\frac{m}{m-k_{\max}}\tilde{\mathbf{A}}_{(i)}-\mathbf{I}\right)\mathcal{P}_{T}\right\|_{(2,2)}\geq\tau\right\}
    \leq 2k_{T}\exp\left(-\frac{m-k_{\max}}{\kappa_{i}k_{T}\mu_{i}}\frac{\tau^{2}}{4(1+\frac{2\tau}{3})}\right), \nonumber
\end{equation}
and
\begin{equation}
    Pr\left\{\left\|\mathcal{P}_{T}\left(\frac{m}{m-k_{\max}}\tilde{\mathbf{A}}_{(i)}\mathbf{\Sigma}_{(i)}^{-1}-\mathbf{\Sigma}_{(i)}^{-1}\right)
    \mathcal{P}_{T}\right\|_{(2,2)}\geq\tau\right\}
    \leq 2k_{T}\exp\left(-\frac{m-k_{\max}}{\kappa_{i}k_{T}\mu_{i}}\frac{\tau^{2}}{4(\kappa_i+\frac{2\tau}{3})}\right). \nonumber
\end{equation}

\end{lemma}

We show the proof of the second part in Appendix \ref{proof_3}.
The first part can be proved in a similar way. Two consequent
corollaries of Lemma \ref{lemma_BEQ_1} show that the restriction
of
$\frac{m}{m-k_{\max}}\textbf{BLKdiag}\left\{\tilde{\mathbf{A}}_{(1)},\cdots,\tilde{\mathbf{A}}_{(L)}\right\}$
to the corresponding support $T$ is near isometric.

\begin{corollary} \label{coro_1}
Denote
$\tilde{\mathbf{A}}_{(i)}=\mathbf{\Sigma}_{(i)}^{-1}\mathbf{A}_{(i)}'\mathcal{P}_{\Omega_{i}^*}\mathbf{A}_{(i)}$.
Given $k_{T}\leq \alpha\frac{m}{L\mu_{\max}\kappa_{\max}\log n}$,
$k_{\max}\leq\gamma m$, and $\frac{1-\gamma}{\alpha}\geq64$, then
with probability at least $1-2n^{-2}$, we have
\begin{equation}\label{BEQ_corollary}
    \left\|
    \textbf{BLKdiag}\left\{\mathcal{P}_{T}\left(\frac{m}{m-k_{\max}}\tilde{\mathbf{A}}_{(1)}-\mathbf{I}\right)\mathcal{P}_{T},
    ~\cdots,~\mathcal{P}_{T}\left(\frac{m}{m-k_{\max}}\tilde{\mathbf{A}}_{(L)}-\mathbf{I}\right)\mathcal{P}_{T}
    \right\}
   \right\|_{(2,2)}<\frac{1}{2}.
\end{equation}
Furthermore, given $k_{T}\leq
\alpha\frac{m}{L\mu_{\max}\kappa_{\max}\log^{2} n}$,
$k_{\max}\leq\gamma m$, and $\frac{1-\gamma}{\alpha}\geq64$, with
at least the same probability, we have
\begin{equation}\label{BEQ_corollary_plus}
\hspace{-1.5em}  \left\|
    \textbf{BLKdiag}\left\{\mathcal{P}_{T}\left(\frac{m}{m-k_{\max}}\tilde{\mathbf{A}}_{(1)}-\mathbf{I}\right)\mathcal{P}_{T},
    ~\cdots,~\mathcal{P}_{T}\left(\frac{m}{m-k_{\max}}\tilde{\mathbf{A}}_{(L)}-\mathbf{I}\right)\mathcal{P}_{T}
    \right\}
   \right\|_{(2,2)}<\frac{1}{2\sqrt{\log n}}.
\end{equation}
\end{corollary}

\begin{IEEEproof}
First, following directly from the first part of Lemma
\ref{lemma_BEQ_1}, for all $i=1,\cdots,L$, it holds
\begin{equation}
    Pr\left\{\left\|\mathcal{P}_{T}\left(\frac{m}{m-k_{\max}}\tilde{\mathbf{A}}_{(i)}
    -\mathbf{I}\right)\mathcal{P}_{T}\right\|_{(2,2)}
    \geq \tau\right\}\leq
    2k_{T}\exp\left\{-\frac{m-k_{\max}}{k_{T}\mu_{\max}\kappa_{\max}}\frac{\tau^{2}}{4(1+\frac{2\tau}{3})}\right\}.
\end{equation}
Taking a union bound over all $i=1,\cdots,L$ yields
\begin{align}\label{e18-2}
    &Pr\left\{\left\|
    \textbf{BLKdiag}\left\{\mathcal{P}_{T}\left(\frac{m}{m-k_{\max}}\tilde{\mathbf{A}}_{(1)}
    -\mathbf{I}\right)\mathcal{P}_{T},
    ~\cdots,~\mathcal{P}_{T}\left(\frac{m}{m-k_{\max}}\tilde{\mathbf{A}}_{(L)}
    -\mathbf{I}\right)\mathcal{P}_{T}
    \right\}
   \right\|_{(2,2)}\geq\tau\right\} \nonumber\\
   &=Pr\left\{\max_{i}\left\{\left\|\mathcal{P}_{T}\left(\frac{m}{m-k_{\max}}\tilde{\mathbf{A}}_{(i)}-\mathbf{I}\right)\mathcal{P}_{T}\right\|_{(2,2)}\right\}\geq \tau\right\}\nonumber\\
   &\leq \sum_{i=1}^L Pr\left\{\left\|\mathcal{P}_{T}\left(\frac{m}{m-k_{\max}}\tilde{\mathbf{A}}_{(i)}
    -\mathbf{I}\right)\mathcal{P}_{T}\right\|_{(2,2)}\geq \tau\right\}\nonumber\\
   &\leq 2k_TL\exp\left\{-\frac{m-k_{\max}}{k_T\mu_{\max}\kappa_{\max}}\frac{\tau^{2}}{4(1+\frac{2\tau}{3})}\right\}.
\end{align}
Plugging in $\tau=\frac{1}{2}$ and using the fact that $k_{T}\leq
\alpha\frac{m}{\mu_{\max}\kappa_{\max}\log n}$ and
$k_{\max}\leq\gamma m$, we get
\begin{align}
    \text{The last line of}~(\ref{e18-2})&=2k_{T}L\exp\left\{-\frac{3(1-\gamma)}{64\alpha}\log n\right\} \nonumber\\
    &=2k_TLn^{-\frac{3(1-\gamma)}{64\alpha}}\nonumber\\
    &\leq 2k_TLn^{-3}\leq2n^{-2},\nonumber
\end{align}
where the first inequality follows from
$\frac{1-\gamma}{\alpha}\geq64$ and the second inequality follows
from $k_TL\leq n$. Similarly, plugging in
$\tau=\frac{1}{2\sqrt{\log n}}$ and using the fact that $k_T\leq
\alpha\frac{m}{\mu_{\max}\kappa_{\max}\log^2 n}$, we prove
\eqref{BEQ_corollary_plus} as long as
$\frac{1-\gamma}{\alpha}\geq64$.
\end{IEEEproof}

\begin{corollary}\label{BEQ_corollary_plusplus}
Given that $k_T\leq \alpha\frac{m}{\mu_{\max}\kappa_{\max}\log
n}$, $k_{\max}\leq\gamma m$, and $\frac{1-\gamma}{\alpha}\geq64$,
then with probability at least $1-2n^{-2}$, we have
\begin{align}\label{e45}
    &\left\|
    \textbf{BLKdiag}\left\{\mathcal{P}_{T}\left(\frac{m}{m-k_{\max}}\tilde{\mathbf{A}}_{(1)}\mathbf{\Sigma}_{(1)}^{-1}-\mathbf{\Sigma}_{(1)}^{-1}\right)\mathcal{P}_{T}, \right.\right. \nonumber\\
    &\left.\left.~\cdots,~\mathcal{P}_{T}\left(\frac{m}{m-k_{\max}}\tilde{\mathbf{A}}_{(L)}\mathbf{\Sigma}_{(L)}^{-1}-\mathbf{\Sigma}_{(L)}^{-1}\right)\mathcal{P}_{T}
    \right\}
   \right\|_{(2,2)}<\frac{\kappa_{\max}}{2}.
\end{align}
\end{corollary}

The proof is almost the same as proving \eqref{BEQ_corollary}
using Lemma \ref{lemma_BEQ_1}. We omit the details for brevity.

Finally, we have the following lemma show that if the support of
the columns in $\mathbf{A}_{(i)}$ is restricted to $\Omega_i^*$,
then no column indexed inside $T$ can be well approximated by the
column indexed outside of $T$. In other words, those columns
correspond to the true signal matrix shall be well distinguished.

\begin{lemma}\label{lemma_BEQ_2}
(Off-support incoherence) Denote
$\tilde{\mathbf{A}}_{(i)}=\mathbf{\Sigma}_{(i)}^{-1}\mathbf{A}_{(i)}'\mathcal{P}_{\Omega_{i}^*}\mathbf{A}_{(i)}$.
Given $k_T\leq \alpha\frac{m}{\mu_{\max}\kappa_{\max}\log n}$ and
$\alpha<\frac{1}{24}$, with probability at least
$1-e^{\frac{1}{4}}n^{-2}$, we have
\begin{align}
    \max_{i\in\{1,\cdots,L\},k\in T^{c}}
    \left\|\mathcal{P}_{T}\tilde{\mathbf{A}}_{(i)}\mathbf{e}_{k}\right\|_2\leq1,
\end{align}
where $\{\mathbf{e}_k\}_{k=1}^n$ is a standard basis in
$\mathcal{R}^n$.
\end{lemma}

The proof of Lemma \ref{lemma_BEQ_2} is given in Appendix
\ref{proof_4}.

With particular note, in the above lemmas and corollaries, all the
requirements on the constants $\alpha$, $\beta$ and $\gamma$
satisfy the bounds in Theorem \ref{recovery_guarantee}.

\section{Exact and Inexact Dual Certificates}
\label{sec:certificate}

This section gives the dual certificates of the RGL model, namely,
the sufficient conditions under which the optimal solution pair of
\eqref{ee3} is unique and equal to the pair of the true signal and
error matrices. Sections \ref{sec:certificate1} and
\ref{sec:certificate2} prove the exact and inexact dual
certificates, respectively. The inexact dual certificate is a
perturbation of the exact one, enabling easier construction of the
certificate in Section \ref{sec:golfing}.

\subsection{Exact Dual Certificate}
\label{sec:certificate1}

Below we show that the optimal solution pair ($\hat{\mathbf{Y}}$,
$\hat{\mathbf{S}}$) of the RGL model \eqref{ee3} is equal to the
true signal and noise pair ($\overline{\mathbf{Y}}$,
$\overline{\mathbf{S}}$) when certain certificate conditions hold.
First we have two preliminary lemmas.

\begin{lemma}\label{equivalence_relation}
Suppose that $\overline{\mathbf{Y}} \in \mathcal{R}^{n \times L}$
and $\overline{\mathbf{S}} \in \mathcal{R}^{m \times L}$ are the
true group sparse signal and sparse error matrices, respectively.
If
$(\overline{\mathbf{Y}}+\mathbf{H},\overline{\mathbf{S}}-\mathbf{F})$
is an optimal solution pair to (\ref{ee3}), where $\mathbf{H}\in
\mathcal{R}^{n\times L}$ and $\mathbf{F}\in\mathcal{R}^{m\times
L}$, then the following results hold:
\begin{enumerate}[i)]
  \item
  $\left[\mathbf{A}_{(1)}\mathbf{h}_{1},\cdots,\mathbf{A}_{(L)}\mathbf{h}_{L}\right]=\mathbf{F}$;
  \item
     $\|\overline{\mathbf{Y}}+\mathbf{H}\|_{2,1}+\lambda\|\overline{\mathbf{S}}-\mathbf{F}\|_{1}
     \geq\|\overline{\mathbf{Y}}\|_{2,1}+\lambda\|\overline{\mathbf{S}}\|_{1}
    +\|\mathcal{P}_{T^{c}}\mathbf{H}\|_{2,1}+\lambda\|\mathcal{P}_{\Omega^{c}}\mathbf{F}\|_1
    +\langle\overline{\mathbf{V}},\mathbf{H}\rangle-\lambda\langle
    \textrm{sgn}(\overline{\mathbf{S}}),\mathbf{F}\rangle.$
\end{enumerate}
where $\overline{\mathbf{V}} \in \mathcal{R}^{n \times L}$
satisfies
$(\mathcal{P}_{T}\overline{\mathbf{V}})^i=\frac{\bar{\mathbf{y}}^i}{\|\bar{\mathbf{y}}^i\|_{2}}$
and $(\mathcal{P}_{T^{c}}\overline{\mathbf{V}})^{i}=\mathbf{0}$,
$\forall i=1,\cdots,n$. Here
$(\mathcal{P}_{T}\overline{\mathbf{V}})^i$ denotes the $i$-th  row
of $\mathcal{P}_{T}\overline{\mathbf{V}}$ and $\bar{\mathbf{y}}^i$
denotes the $i$-th row of $\overline{\mathbf{Y}}$.
\end{lemma}

The proof of Lemma \ref{equivalence_relation} is given in Appendix
\ref{appendix_equivalence_relation}.

\begin{lemma}\label{null_space_lemma}
For any two matrices $\mathbf{H}\in \mathcal{R}^{n\times L}$ and
$\mathbf{F}\in\mathcal{R}^{m\times L}$, with probability at least
$1-2n^{-2}$,
$\mathbf{H}=\mathbf{0}$ and $\mathbf{F}=\mathbf{0}$ if the
following conditions are satisfied:
\begin{enumerate}[i)]
  \item $k_{T}\leq \alpha\frac{m}{\mu_{\max}\kappa_{\max}\log^2 n}$, $k_{\max}\leq\gamma m$, and
  $\frac{1-\gamma}{\alpha}\geq64$;
  \item $\left[\mathbf{A}_{(1)}\mathbf{h}_{1},\cdots,\mathbf{A}_{(L)}\mathbf{h}_{L}\right]=\mathbf{F};$
  \item $\mathcal{P}_{T^c}\mathbf{H}=\mathbf{0}$ and $\mathcal{P}_{\Omega^c}\mathbf{F}=\mathbf{0}.$
\end{enumerate}
\end{lemma}

The proof of Lemma \ref{null_space_lemma} is given in Appendix
\ref{appendix_null_space}.

\begin{theorem} \label{theorem_exact_duality}
(\textit{Exact Duality}) Suppose $\overline{\mathbf{Y}} \in
\mathcal{R}^{n \times L}$ and $\overline{\mathbf{S}} \in \mathcal{R}^{m
\times L}$ are the true group sparse signal and sparse error
matrices satisfying the assumptions in Theorem \ref{recovery_guarantee}. The pair ($\overline{\mathbf{Y}}$, $\overline{\mathbf{S}}$) is
the unique solution to the RGL model (\ref{ee3}) with a high probability if there exists a
dual certificate $\mathbf{W} \in \mathcal{R}^{m \times L}$ such
that
\begin{align}
    &\mathcal{P}_{T}\left[\mathbf{A}_{(1)}'\mathbf{w}_{1}, \cdots,\mathbf{A}_{(L)}'\mathbf{w}_{L}\right] = \overline{\mathbf{V}}, \label{e48-1}\\
    &\left\|\left[\mathbf{A}_{(1)}'\mathbf{w}_{1},\cdots,\mathbf{A}_{(L)}'\mathbf{w}_{L}\right]\right\|_{2,\infty} < 1, \label{e48-2}\\
    &\mathcal{P}_{\Omega}\mathbf{W}=\lambda \textrm{sgn}(\overline{\mathbf{S}}), \label{e48-3} \\
    &\|\mathbf{W}\|_{\infty}<\lambda, \label{e48-4}
\end{align}
where $\overline{\mathbf{V}} \in \mathcal{R}^{n \times L}$
satisfies
$(\mathcal{P}_{T}\overline{\mathbf{V}})^i=\frac{\bar{\mathbf{y}}^i}{\|\bar{\mathbf{y}}^i\|_{2}}$
and $(\mathcal{P}_{T^{c}}\overline{\mathbf{V}})^{i}=\mathbf{0}$,
$\forall i=1,\cdots,n$. Here
$(\mathcal{P}_{T}\overline{\mathbf{V}})^i$ denotes the $i$-th  row
of $\mathcal{P}_{T}\overline{\mathbf{V}}$ and $\bar{\mathbf{y}}^i$
denotes the $i$-th row of $\overline{\mathbf{Y}}$.
\end{theorem}


\begin{IEEEproof}
Suppose that
$(\overline{\mathbf{Y}}+\mathbf{H},\overline{\mathbf{S}}-\mathbf{F})$,
where $\mathbf{H}\in \mathcal{R}^{n\times L}$ and
$\mathbf{F}\in\mathcal{R}^{m\times L}$, is an optimal solution
pair to (\ref{ee3}). Therefore, the two results in Lemma
\ref{equivalence_relation} hold true. Proving that the pair
($\overline{\mathbf{Y}}$, $\overline{\mathbf{S}}$) is the unique
solution to (\ref{ee3}) is equivalent to showing that
$\mathbf{H}=\mathbf{0}$ and $\mathbf{F}=\mathbf{0}$. Hence, the
proof resorts to verifying the three conditions in Lemma
\ref{null_space_lemma}.

Since $\overline{\mathbf{Y}}$ and $\overline{\mathbf{S}}$ satisfy
the assumptions in Theorem \ref{recovery_guarantee}, we have
$$k_{T}\leq\alpha\frac{m}{\mu_{\max}\kappa_{\max}\log^{2} n}, \quad k_{\max}\leq\gamma
\frac{m}{\kappa_{\max}}, \quad \alpha \leq \frac{1}{9600}, \quad
\gamma \leq \frac{1}{4}.$$ Considering $\kappa_{max} \geq 1$, we
know that condition i) in Lemma \ref{null_space_lemma} holds. By
result i) of Lemma \ref{equivalence_relation}, condition ii) also
holds. Therefore, it remains to verify condition iii), namely,
$\mathcal{P}_{T^c}\mathbf{H}=\mathbf{0}$ and
$\mathcal{P}_{\Omega^c}\mathbf{F}=\mathbf{0}$.

Consider the term
$\langle\overline{\mathbf{V}},\mathbf{H}\rangle-\lambda\langle
\textrm{sgn}(\overline{\mathbf{S}}),\mathbf{F}\rangle$ at the
right-hand side of result ii) in Lemma \ref{equivalence_relation}.
From (\ref{e48-1}), it follows
\[\overline{\mathbf{V}} =
\left[\mathbf{A}_{(1)}'\mathbf{w}_{1},\cdots,\mathbf{A}_{(L)}'\mathbf{w}_{L}\right]
-
\mathcal{P}_{T^c}\left[\mathbf{A}_{(1)}'\mathbf{w}_{1},\cdots,\mathbf{A}_{(L)}'\mathbf{w}_{L}\right],\]
and consequently
\begin{align}\label{e55-1}
    \langle\overline{\mathbf{V}},\mathbf{H}\rangle
=\left\langle\left[\mathbf{A}_{(1)}'\mathbf{w}_{1},\cdots,\mathbf{A}_{(L)}'\mathbf{w}_{L}\right],
\mathbf{H}\right\rangle - \left\langle
\mathcal{P}_{T^c}\left[\mathbf{A}_{(1)}'\mathbf{w}_{1},\cdots,\mathbf{A}_{(L)}'\mathbf{w}_{L}\right],
      \mathbf{H}\right\rangle.
\end{align}
By adjoint relation
$\left\langle\left[\mathbf{A}_{(1)}'\mathbf{w}_{1},\cdots,\mathbf{A}_{(L)}'\mathbf{w}_{L}\right],
\mathbf{H}\right\rangle =
\langle\mathbf{W},\left[\mathbf{A}_{(1)}\mathbf{h}_{1}, \cdots,
\mathbf{A}_{(L)}\mathbf{h}_{L}\right]\rangle$ and the fact
$\mathbf{F}=\left[\mathbf{A}_{(1)}\mathbf{h}_{1},\cdots,\mathbf{A}_{(L)}\mathbf{h}_{L}\right]$,
(\ref{e55-1}) gives
\begin{align}\label{e55-2}
    \langle\overline{\mathbf{V}},\mathbf{H}\rangle
=\langle\mathbf{W},\mathbf{F}\rangle - \left\langle
\mathcal{P}_{T^c}\left[\mathbf{A}_{(1)}'\mathbf{w}_{1},\cdots,\mathbf{A}_{(L)}'\mathbf{w}_{L}\right],
      \mathbf{H}\right\rangle.
\end{align}
On the other hand, from (\ref{e48-3}),
$\mathcal{P}_{\Omega}\mathbf{W}=\lambda \textrm{sgn}(\overline{\mathbf{S}})$ and
hence $\lambda \textrm{sgn}(\overline{\mathbf{S}}) = \mathbf{W} -
\mathcal{P}_{\Omega^c}\mathbf{W}$. Therefore we have
\begin{align}\label{e55-3}
\lambda\langle \textrm{sgn}(\overline{\mathbf{S}}),\mathbf{F}\rangle =
\langle\mathbf{W},\mathbf{F}\rangle -
\langle\mathcal{P}_{\Omega^c}\mathbf{W},\mathbf{F}\rangle.
\end{align}
Combining (\ref{e55-2}) and (\ref{e55-3}) yields
\begin{align}\label{e55}
    \langle\overline{\mathbf{V}},\mathbf{H}\rangle-\lambda\langle \textrm{sgn}(\overline{\mathbf{S}}),\mathbf{F}\rangle
    =\langle\mathcal{P}_{\Omega^c}\mathbf{W},\mathbf{F}\rangle - \left\langle
\mathcal{P}_{T^c}\left[\mathbf{A}_{(1)}'\mathbf{w}_{1},\cdots,\mathbf{A}_{(L)}'\mathbf{w}_{L}\right],
      \mathbf{H}\right\rangle.
\end{align}
Substituting (\ref{e55}) into result ii) of Lemma
\ref{equivalence_relation} gives
\begin{align}\label{e56-1}
    & \|\overline{\mathbf{Y}}+\mathbf{H}\|_{2,1}+\lambda\|\overline{\mathbf{S}}-\mathbf{F}\|_{1} \geq\|\overline{\mathbf{Y}}\|_{2,1}+\lambda\|\overline{\mathbf{S}}\|_{1} \nonumber \\
    & +\|\mathcal{P}_{T^{c}}\mathbf{H}\|_{2,1} - \left\langle
\mathcal{P}_{T^c}\left[\mathbf{A}_{(1)}'\mathbf{w}_{1},\cdots,\mathbf{A}_{(L)}'\mathbf{w}_{L}\right],
      \mathbf{H}\right\rangle +
      \lambda\|\mathcal{P}_{\Omega^{c}}\mathbf{F}\|_1+\langle\mathcal{P}_{\Omega^c}\mathbf{W},\mathbf{F}\rangle.
\end{align}

From
$\|\left[\mathbf{A}_{(1)}'\mathbf{w}_{1},\cdots,\mathbf{A}_{(L})'\mathbf{w}_{L}\right]\|_{2,\infty}
< 1$ in (\ref{e48-2}), we know that
\begin{align*}
&- \left\langle
\mathcal{P}_{T^c}\left[\mathbf{A}_{(1)}'\mathbf{w}_{1},\cdots,\mathbf{A}_{(L)}'\mathbf{w}_{L}\right],
      \mathbf{H}\right\rangle \\
\geq& -
\left\|\left[\mathbf{A}_{(1)}'\mathbf{w}_{1},\cdots,\mathbf{A}_{(L)}'\mathbf{w}_{L}\right]\right\|_{2,\infty}
\|\mathcal{P}_{T^c} \mathbf{H}\|_{2,1} \\
\geq& - \|\mathcal{P}_{T^c}\mathbf{H}\|_{2,1},
\end{align*}
where both inequalities turn to equalities if and only if
$\|\mathcal{P}_{T^c} \mathbf{H}\|_{2,1} = 0$. From
$\|\mathbf{W}\|_{\infty}<\lambda$ in (\ref{e48-4}), we know that
$$\langle\mathcal{P}_{\Omega^c}\mathbf{W},\mathbf{F}\rangle \geq -
\|\mathbf{W}\|_{\infty} \|\mathcal{P}_{\Omega^{c}}\mathbf{F}\|_1
\geq - \lambda \|\mathcal{P}_{\Omega^{c}}\mathbf{F}\|_1,$$
where both inequalities turn to equalities turns to equality if and only if
$\|\mathcal{P}_{\Omega^{c}}\mathbf{F}\|_1=0$. Therefore,
(\ref{e56-1}) leads to
\begin{align}\label{e56}
    & \|\overline{\mathbf{Y}}+\mathbf{H}\|_{2,1}+\lambda\|\overline{\mathbf{S}}-\mathbf{F}\|_{1}
    \geq
    \|\overline{\mathbf{Y}}\|_{2,1}+\lambda\|\overline{\mathbf{S}}\|_{1},
\end{align}
where the inequality turns to an equality if and only if
$\|\mathcal{P}_{T^c} \mathbf{H}\|_{2,1} = 0$ and
$\|\mathcal{P}_{\Omega^{c}}\mathbf{F}\|_1=0$.

Since by hypothesis
$(\overline{\mathbf{Y}}+\mathbf{H},\overline{\mathbf{S}}-\mathbf{F})$
is the optimal solution pair, the inequality in \eqref{e56} must
be an equality. Therefore, it follows that $\|\mathcal{P}_{T^c}
\mathbf{H}\|_{2,1} = 0$ and
$\|\mathcal{P}_{\Omega^{c}}\mathbf{F}\|_1=0$, which complete the
proof.
\end{IEEEproof}

It is generally difficult to directly construct an exact dual
certificate. Thus, following \cite{candes2010} and
\cite{gross2012}, we perturb the constraints
\eqref{e48-1}-\eqref{e48-4} by making \eqref{e48-2} and
\eqref{e48-4} more stringent, which in turn relaxes \eqref{e48-1}
and \eqref{e48-3}.

\subsection{Inexact Dual Certificate}
\label{sec:certificate2}

\begin{theorem} \label{theorem_inexact_duality}
\textit{(Inexact Duality)} Suppose that $\overline{\mathbf{Y}} \in
\mathcal{R}^{n \times L}$ and $\overline{\mathbf{S}} \in
\mathcal{R}^{m \times L}$ are the true group sparse signal and
sparse error matrices satisfying the assumptions in Theorem
\ref{recovery_guarantee}. The pair ($\overline{\mathbf{Y}}$,
$\overline{\mathbf{S}}$) is the unique solution to the RGL model
(\ref{ee3}) if the parameter $\lambda<1$ and there exists a dual
certificate $(\mathbf{W}, \mathbf{V}) \in \mathcal{R}^{m \times L}
\times \mathcal{R}^{n \times L}$ such that
\begin{align}
&\|\mathcal{P}_{T}\mathbf{V}-\overline{\mathbf{V}}\|_{F}\leq\frac{\lambda}{4\sqrt{\kappa_{\max}}},\label{inexact_1_1}\\
&\|\mathcal{P}_{T^{c}}\mathbf{V}\|_{2,\infty}\leq\frac{1}{4}, \label{inexact_1_2}\\
&\|\mathcal{P}_{\Omega^{c}}\mathbf{W}\|_{\infty}\leq\frac{\lambda}{4},
\label{inexact_1_3}
\end{align}
and
\begin{equation}\label{inexact_2}
\mathbf{V}=\left[\mathbf{A}_{(1)}'\mathcal{P}_{\Omega_{1}^{c}}\mathbf{w}_{1},
\cdots,
\mathbf{A}_{(L)}'\mathcal{P}_{\Omega_{L}^{c}}\mathbf{w}_{L}\right]
+\lambda\left[\mathbf{A}_{(1)}'\textrm{sgn}(\bar{\mathbf{s}}_{1}),
\cdots,
\mathbf{A}_{(L)}'\textrm{sgn}(\bar{\mathbf{s}}_{L})\right],
\end{equation}
where $\overline{\mathbf{V}} \in \mathcal{R}^{n \times L}$ satisfies
$(\mathcal{P}_{T}\overline{\mathbf{V}})^i=\frac{\bar{\mathbf{y}}^i}{\|\bar{\mathbf{y}}^i\|_{2}}$
and $(\mathcal{P}_{T^{c}}\overline{\mathbf{V}})^{i}=\mathbf{0}$.
\end{theorem}

\begin{IEEEproof}
Suppose that
$(\overline{\mathbf{Y}}+\mathbf{H},\overline{\mathbf{S}}-\mathbf{F})$
is an optimal solution pair to (\ref{ee3}). As discussed in the
proof of Theorem \ref{theorem_exact_duality}, it is enough to show
that $\mathcal{P}_{T^c}\mathbf{H}=\mathbf{0}$ and
$\mathcal{P}_{\Omega^c}\mathbf{F}=\mathbf{0}$.

Consider the term
$\langle\overline{\mathbf{V}},\mathbf{H}\rangle-\lambda\langle
\textrm{sgn}(\overline{\mathbf{S}}),\mathbf{F}\rangle$ in result
ii) of Lemma \ref{equivalence_relation}. Using the equation
$\mathbf{V} = \mathcal{P}_{T}\mathbf{V} +
\mathcal{P}_{T^c}\mathbf{V}$, we rewrite the term as
\begin{align}\label{interim_inexact_1}
    \langle\overline{\mathbf{V}},\mathbf{H}\rangle-\lambda\langle
\textrm{sgn}(\overline{\mathbf{S}}),\mathbf{F}\rangle =
    \langle\overline{\mathbf{V}} -  \mathcal{P}_{T}\mathbf{V},\mathbf{H}\rangle - \langle\mathcal{P}_{T^c}\mathbf{V},\mathbf{H}\rangle + \langle\mathbf{V},\mathbf{H}\rangle-\lambda\langle
\textrm{sgn}(\overline{\mathbf{S}}),\mathbf{F}\rangle.
\end{align}
Consider the term
$\langle\mathbf{V},\mathbf{H}\rangle-\lambda\langle
\textrm{sgn}(\overline{\mathbf{S}}),\mathbf{F}\rangle$ on the
right hand side of \eqref{interim_inexact_1}. By
\eqref{inexact_2}, we have
\begin{align*}
    &\langle\mathbf{V}, \mathbf{H}\rangle-\lambda\langle
     \textrm{sgn}(\overline{\mathbf{S}}),\mathbf{F}\rangle \\
   =& \left\langle\left[\mathbf{A}_{(1)}'\mathcal{P}_{\Omega_{1}^{c}}\mathbf{w}_{1},
\cdots,
\mathbf{A}_{(L)}'\mathcal{P}_{\Omega_{L}^{c}}\mathbf{w}_{L}\right],\mathbf{H}\right\rangle \\
    &+\lambda\left\langle\left[\mathbf{A}_{(1)}'\textrm{sgn}(\bar{\mathbf{s}}_{1}),
\cdots,
\mathbf{A}_{(L)}'\textrm{sgn}(\bar{\mathbf{s}}_{L})\right],\mathbf{H}\right\rangle
-\lambda\langle\textrm{sgn}(\overline{\mathbf{S}}),\mathbf{F}\rangle.
\end{align*}
By adjoint relations of inner products, we have
\begin{align*}
&\left\langle\left[\mathbf{A}_{(1)}'\mathcal{P}_{\Omega_{1}^{c}}\mathbf{w}_{1},
\cdots,
\mathbf{A}_{(L)}'\mathcal{P}_{\Omega_{L}^{c}}\mathbf{w}_{L}\right],\mathbf{H}\right\rangle\\
=&\left\langle\left[\mathcal{P}_{\Omega_{1}^{c}}\mathbf{A}_{(1)}\mathbf{h}_{1},
\cdots,
\mathcal{P}_{\Omega_{L}^{c}}\mathbf{A}_{(L)}\mathbf{h}_{L}\right],\mathbf{W}\right\rangle,
\end{align*}
and
\begin{align*}
&\langle\left[\mathbf{A}_{(1)}'\textrm{sgn}(\bar{\mathbf{s}}_{1}),
\cdots,
\mathbf{A}_{(L)}'\textrm{sgn}(\bar{\mathbf{s}}_{L})\right],\mathbf{H}\rangle\\
=&\langle \textrm{sgn}(\overline{\mathbf{S}}),
\left[\mathbf{A}_{(1)}\mathbf{h}_{1}, \cdots,
    \mathbf{A}_{(L)}\mathbf{h}_{L}\right]\rangle.
\end{align*}
Thus, it holds
\begin{align*}
      &\langle\mathbf{V}, \mathbf{H}\rangle-\lambda\langle\textrm{sgn}(\overline{\mathbf{S}}),\mathbf{F}\rangle\\
    = & \left\langle\left[\mathcal{P}_{\Omega_{1}^{c}}\mathbf{A}_{(1)}\mathbf{h}_{1},
\cdots,
\mathcal{P}_{\Omega_{L}^{c}}\mathbf{A}_{(L)}\mathbf{h}_{L}\right],\mathbf{W}\right\rangle \\
      &+\lambda\left\langle \textrm{sgn}(\overline{\mathbf{S}}), \left[\mathbf{A}_{(1)}\mathbf{h}_{1}, \cdots,
      \mathbf{A}_{(L)}\mathbf{h}_{L}\right]\right\rangle
      -\lambda\langle\textrm{sgn}(\overline{\mathbf{S}}),\mathbf{F}\rangle.
\end{align*}
According to conclusion i) in Lemma \ref{equivalence_relation},
which is
$\left[\mathbf{A}_{(1)}\mathbf{h}_{1},\cdots,\mathbf{A}_{(L)}\mathbf{h}_{L}\right]=\mathbf{F}$,
it follows
\begin{align*}
    \langle\mathbf{V}, \mathbf{H}\rangle-\lambda\langle \textrm{sgn}(\overline{\mathbf{S}}, \mathbf{F})\rangle
    =  \langle
    \mathcal{P}_{\Omega^{c}}\mathbf{F},\mathbf{W}\rangle.
\end{align*}
Combining \eqref{interim_inexact_1} and above equality gives
\begin{align}\label{interim_inexact_2}
& \langle\overline{\mathbf{V}},\mathbf{H}\rangle - \lambda\langle
\textrm{sgn}(\overline{\mathbf{S}}), \mathbf{F}\rangle  = \langle\overline{\mathbf{V}} -
\mathcal{P}_{T}\mathbf{V},\mathbf{H}\rangle -
\langle\mathcal{P}_{T^c}\mathbf{V},\mathbf{H}\rangle +
 \langle \mathcal{P}_{\Omega^{c}}\mathbf{F},\mathbf{W}\rangle.
\end{align}

Next, we manage to find out a lower bound for the right-hand side
of the equality \eqref{interim_inexact_2}. First, by
\eqref{inexact_1_1},
\[\langle\overline{\mathbf{V}} -
\mathcal{P}_{T}\mathbf{V},\mathbf{H}\rangle
\geq
-\|\mathcal{P}_{T}\mathbf{V} - \mathbf{V}_{0} \|_F
\|\mathcal{P}_{T} \mathbf{H}\|_F
\geq
-\frac{\lambda}{4\sqrt{\kappa_{\max}}} \|\mathcal{P}_{T}
\mathbf{H}\|_F.\]
Then, by \eqref{inexact_1_2},
\[- \langle\mathcal{P}_{T^c}\mathbf{V},\mathbf{H}\rangle
\geq
-
\|\mathcal{P}_{T^c}\mathbf{V}\|_{2,\infty}\|\mathcal{P}_{T^c}
\mathbf{H}\|_{2,1}
\geq
-\frac{1}{4}\|\mathcal{P}_{T^c} \mathbf{H}\|_{2,1}.\]
Finally, by \eqref{inexact_1_3},
\[\langle \mathcal{P}_{\Omega^{c}}\mathbf{F},\mathbf{W}\rangle
\geq
 - \|\mathcal{P}_{\Omega^{c}}\mathbf{F}\|_1
\|\mathcal{P}_{\Omega^{c}} \mathbf{W}\|_\infty \geq -
\frac{\lambda}{4} \|\mathcal{P}_{\Omega^{c}}\mathbf{F}\|_1.\]
Therefore, \eqref{interim_inexact_2} gives
\begin{align*}
\langle\overline{\mathbf{V}},\mathbf{H}\rangle - \lambda\langle
\textrm{sgn}(\overline{\mathbf{S}}), \mathbf{F}\rangle
\geq -
\frac{\lambda}{4\sqrt{\kappa_{\max}}} \|\mathcal{P}_{T}
\mathbf{H}\|_F  - \frac{1}{4}\|\mathcal{P}_{T^c}
\mathbf{H}\|_{2,1} - \frac{\lambda}{4}
\|\mathcal{P}_{\Omega^{c}}\mathbf{F}\|_1.
\end{align*}
Substitute the above inequality into conclusion ii) of Lemma
\ref{equivalence_relation} gives
\begin{align*}
    & \|\overline{\mathbf{Y}}+\mathbf{H}\|_{2,1}+\lambda\|\overline{\mathbf{S}}-\mathbf{F}\|_{1}  \\
\geq& \|\overline{\mathbf{Y}}\|_{2,1}+\lambda\|\overline{\mathbf{S}}\|_{1}
    + \frac{3}{4} \|\mathcal{P}_{T^{c}}\mathbf{H}\|_{2,1}+\frac{3\lambda}{4}\|\mathcal{P}_{\Omega^{c}}\mathbf{F}\|_1
    -\frac{\lambda}{4\sqrt{\kappa_{\max}}} \|\mathcal{P}_{T} \mathbf{H}\|_F.
\end{align*}
Since
$(\overline{\mathbf{Y}}+\mathbf{H},\overline{\mathbf{S}}-\mathbf{F})$
is an optimal solution pair to \eqref{ee3}, it follows that
$\|\overline{\mathbf{Y}}+\mathbf{H}\|_{2,1}+\lambda\|\overline{\mathbf{S}}-\mathbf{F}\|_{1}
\leq
\|\overline{\mathbf{Y}}\|_{2,1}+\lambda\|\overline{\mathbf{S}}\|_{1}$.
Hence, we have
\begin{align}\label{interim_inexact_3}
& \frac{3}{4}
  \|\mathcal{P}_{T^{c}}\mathbf{H}\|_{2,1}+\frac{3\lambda}{4}\|\mathcal{P}_{\Omega^{c}}\mathbf{F}\|_1
    -\frac{\lambda}{4\sqrt{\kappa_{\max}}} \|\mathcal{P}_{T} \mathbf{H}\|_F \leq 0.
\end{align}
To complete the proof, we need to show that inequality
\eqref{interim_inexact_3} implies
$\mathcal{P}_{T^c}\mathbf{H}=\mathbf{0}$ and
$\mathcal{P}_{\Omega^c}\mathbf{F}=\mathbf{0}$. The proof is given
in Appendix \ref{appendix_inexact_duality} using the concentration
inequalities in Section \ref{sec:main3}.
\end{IEEEproof}

\section{Construction of Dual Certificate}
\label{sec:golfing}

From Theorem \ref{theorem_inexact_duality}, we know that proving
following the theorem is sufficient for proving Theorem
\ref{recovery_guarantee}.

\begin{theorem}\label{theorem_dual_certificate}
Under the assumptions in Theorem \ref{recovery_guarantee}, with
a high probability, there exists a pair of dual certificate
$(\mathbf{U}, \mathbf{W})$ such that
\[\mathbf{U}=\left[\mathbf{A}_{(1)}'\mathcal{P}_{\Omega_{1}^{c}}\mathbf{w}_{1},\cdots,\mathbf{A}_{(L)}'\mathcal{P}_{\Omega_{L}^{c}}\mathbf{w}_{L}\right],\]
and
\begin{align}
    &\left\|\lambda\mathcal{P}_{T^{c}}\left[\mathbf{A}_{(1)}'\textrm{sgn}(\bar{\mathbf{s}}_{1}),\cdots,\mathbf{A}_{(L)}'\textrm{sgn}(\bar{\mathbf{s}}_{L})\right]\right\|_{2,\infty}\leq\frac{1}{8}, \label{cetificate1}\\
    &\left\|\mathcal{P}_{T}\mathbf{U} +\lambda\mathcal{P}_{T}\left[\mathbf{A}_{(1)}'\textrm{sgn}(\bar{\mathbf{s}}_{1}),\cdots,\mathbf{A}_{(L)}'\textrm{sgn}(\bar{\mathbf{s}}_{L})\right]
-\overline{\mathbf{V}}\right\|_{F}\leq\frac{\lambda}{4\sqrt{\kappa_{\max}}}, \label{cetificate2}\\
    &\|\mathcal{P}_{T^{c}}\mathbf{U}\|_{2,\infty}\leq\frac{1}{8},  \label{cetificate3}\\
    &\|\mathcal{P}_{\Omega^{c}}\mathbf{W}\|_{\infty}\leq\frac{\lambda}{4}, \label{cetificate4}
\end{align}
where $\overline{\mathbf{V}} \in \mathcal{R}^{n \times L}$
satisfies
$(\mathcal{P}_{T}\overline{\mathbf{V}})^i=\frac{\bar{\mathbf{y}}^i}{\|\bar{\mathbf{y}}^i\|_{2}}$
and $(\mathcal{P}_{T^{c}}\overline{\mathbf{V}})^{i}=\mathbf{0}$,
$\forall i=1,\cdots,n$.
\end{theorem}

Comparing to Theorem \ref{theorem_inexact_duality}, Theorem
\ref{theorem_dual_certificate} breaks
$\|\mathcal{P}_{T^{c}}\mathbf{V}\|_{2,\infty}\leq\frac{1}{4}$ in
\eqref{inexact_1_2} into two constraints \eqref{cetificate1} and
\eqref{cetificate3}. Thus, Theorem \ref{theorem_dual_certificate}
implies that an inexact dual certificate exists with high
probability. Therefore, Theorem \ref{recovery_guarantee} holds
true according to Theorem \ref{theorem_inexact_duality}.

The construction procedure follows the golfing scheme (see
\cite{li2012}, \cite{candes2009}, and \cite{candes2010}).
Basically, it constructs a sequence of matrices
$\{\mathbf{Q}_{(j)}\}_{j=0}^l$ via $l$ sampled batches of row
vectors in each
$\mathcal{P}_{\Omega_i^*}\mathbf{A}_{(i)},~i\in\{1,\cdots,L\}$, so
that different batches are not overlapped and the sequence
$\left\{\|\mathbf{Q}_{(j)}\|_F\right\}_{j=0}^l$ shrinks
exponentially fast in finite steps with a high probability. We then
write $\mathbf{W}$ and subsequently $\mathbf{U}$ as functions of
$\{\mathbf{Q}_{(j)}\}_{j=0}^l$ so that they meet the constraints
\eqref{cetificate1}-\eqref{cetificate4}.

Define the initial value of the sequence
$\{\mathbf{Q}_{(j)}\}_{j=0}^l$ as
\begin{equation}\label{definition_Q_0}
    \mathbf{Q}_{(0)}=\overline{\mathbf{V}}-\lambda\mathcal{P}_{T}[\mathbf{A}_{(1)}'\textrm{sgn}(\bar{\mathbf{s}}_{1})
                      ,\cdots,\mathbf{A}_{(L)}'\textrm{sgn}(\bar{\mathbf{s}}_{L})].
\end{equation}
For each $i=1,\cdots,L$, we split the maximal non-corrupted set
$\Omega_{i}^*$ into $l$ disjoint batch sets, namely,
$\Omega_{i}^*\supseteq K_{i1}\bigcup\cdots\bigcup K_{il}$, so that
for any $j=1,\cdots,l$, the cardinalities of the sets $|K_{ij}|$
satisfy $|K_{1j}|=\cdots=|K_{Lj}|\triangleq m_j$. Notice that it
is possible to split $\Omega_{i}^*$ in this way since we enforce
$|\Omega_{1}^*|=\cdots=|\Omega_{L}^*|=m-k_{\max}$.

Define
$\tilde{\mathbf{A}}_{(i,j)}=\mathbf{\Sigma}_{(i)}^{-1}\mathbf{A}_{(i)}'\mathcal{P}_{K_{ij}}\mathbf{A}_{(i)}$
and the total number of batches $l \triangleq \lfloor\log
n+1\rfloor$. For each $j=1,\cdots,l$, recursively define
\begin{align}\label{definition_recursive}
    \mathbf{Q}_{(j)}=&\left[\mathcal{P}_{T}\left(\mathbf{I}-\frac{m}{m_{j}}\tilde{\mathbf{A}}_{(1,j)}\right)\mathcal{P}_{T}\mathbf{q}_{(j-1)1}
    ,\cdots
    ,\mathcal{P}_{T}\left(\mathbf{I}-\frac{m}{m_{j}}\tilde{\mathbf{A}}_{(L,j)}\right)\mathcal{P}_{T}\mathbf{q}_{(j-1)L}\right] \nonumber\\
    =&\left[\left(\prod_{r=1}^{j}\mathcal{P}_{T}\left(\mathbf{I}-\frac{m}{m_{r}}\tilde{\mathbf{A}}_{(1,r)}\right)\mathcal{P}_{T}\right)\mathbf{q}_{(0)1}
    ,\cdots
    ,\left(\prod_{r=1}^{j}\mathcal{P}_{T}\left(\mathbf{I}-\frac{m}{m_{r}}\tilde{\mathbf{A}}_{(L,r)}\right)\mathcal{P}_{T}\right)\mathbf{q}_{(0)L}\right].
\end{align}
Furthermore, notice that
$k_{\max}\leq\gamma\frac{m}{\kappa_{\max}}$ with
$\gamma\leq\frac{1}{4}$. We choose
\[m_1=m_2=\frac{m}{4}, \quad m_j= \frac{m}{4\log n},\forall
j\geq3.\]

The following lemma shows that
$\left\{\|\mathbf{Q}_{(j)}\right\|_F\}_{j=0}^l$ shrinks
exponentially fast with a high probability.
\begin{lemma}\label{batch_shrink_lemma}
Given $k_T\leq\alpha\frac{m}{\mu_{\max}\kappa_{\max}\log^2n}$ with
$\alpha\leq\frac{1}{256}$, then, with probability at least
$1-2n^{-1}$, the following set of inequalities hold simultaneously
\begin{equation}\label{batch_inequalities}
    \left\|\textbf{BLKdiag}\left\{\mathcal{P}_{T}\left(\frac{m}{m_{j}}\tilde{\mathbf{A}}_{(1,j)}-\mathbf{I}\right)\mathcal{P}_{T},
    ~\cdots,
    ~\mathcal{P}_{T}\left(\frac{m}{m_{j}}\tilde{\mathbf{A}}_{(L,j)}-\mathbf{I}\right)\mathcal{P}_{T}\right\}\right\|_{(2,2)}\leq c_{j},
\end{equation}
where $c_{1}=c_{2}=\frac{1}{2\sqrt{\log n}}$ and $c_{j}=\frac{1}{2},~j\geq 3$.
\end{lemma}

\begin{proof}
Following the proof of Lemma \ref{lemma_BEQ_1}, for any
$i=1,\cdots,L$ and $j=1,\cdots,l$ we have
\begin{equation}
    Pr\left\{\left\|\mathcal{P}_{T}\left(\frac{m}{m_j}\tilde{\mathbf{A}}_{(i,j)}-\mathbf{I}\right)\mathcal{P}_{T}\right\|_{(2,2)}\geq\tau\right\}
    \leq 2k_{T}\exp\left(-\frac{m_j}{\kappa_{i}k_{T}\mu_{i}}\frac{\tau^{2}}{4(1+\frac{2\tau}{3})}\right). \nonumber
\end{equation}
Next, same as the proof of \eqref{BEQ_corollary} and
\eqref{BEQ_corollary_plus}, for each $j=1,\cdots,l$, taking a
union bound over all $i=1,\cdots,L$, which gives
\begin{align}
  &Pr\left\{\left\|
    \textbf{BLKdiag}\left\{\mathcal{P}_{T}\left(\frac{m}{m_j}\tilde{\mathbf{A}}_{(1,j)}
    -\mathbf{I}\right)\mathcal{P}_{T},
    ~\cdots,~\mathcal{P}_{T}\left(\frac{m}{m_j}\tilde{\mathbf{A}}_{(L,j)}
    -\mathbf{I}\right)\mathcal{P}_{T}
    \right\}
   \right\|_{(2,2)}\geq\tau\right\} \nonumber\\
   \leq& 2k_TL\exp\left\{-\frac{m_j}{k_T\mu_{\max}\kappa_{\max}}\frac{\tau^{2}}{4(1+\frac{2\tau}{3})}\right\}.\label{interim_batch}
\end{align}
If $j\geq3$, then substituting $\tau=\frac{1}{2}$ and
$m_j=\frac{m}{4\log n}$ into above inequality gives
\begin{align*}
&Pr\left\{\left\|
    \textbf{BLKdiag}\left\{\mathcal{P}_{T}\left(\frac{m}{m_j}\tilde{\mathbf{A}}_{(1,j)}
    -\mathbf{I}\right)\mathcal{P}_{T},
    ~\cdots,~\mathcal{P}_{T}\left(\frac{m}{m_j}\tilde{\mathbf{A}}_{(L,j)}
    -\mathbf{I}\right)\mathcal{P}_{T}
    \right\}
   \right\|_{(2,2)}\geq\tau\right\} \nonumber\\
   \leq& 2k_TL\exp\left\{-\frac{3}{256}\frac{m}{k_T\mu_{\max}\kappa_{\max}\log n}\right\}\nonumber\\
   \leq& 2k_TL\exp\left\{-3\log n\right\}\leq2n^{-2},
\end{align*}
where the second inequality follows from
$k_T\leq\alpha\frac{m}{\mu_{\max}\kappa_{\max}\log^2n}$ and
$\alpha\leq\frac{1}{256}$. If $j\leq2$, then substituting
$\tau=\frac{1}{2\sqrt{\log n}}$ and $m_j=\frac{m}{4}$ into
\eqref{interim_batch} gives
\begin{align*}
&Pr\left\{\left\|
    \textbf{BLKdiag}\left\{\mathcal{P}_{T}\left(\frac{m}{m_j}\tilde{\mathbf{A}}_{(1,j)}
    -\mathbf{I}\right)\mathcal{P}_{T},
    ~\cdots,~\mathcal{P}_{T}\left(\frac{m}{m_j}\tilde{\mathbf{A}}_{(L,j)}
    -\mathbf{I}\right)\mathcal{P}_{T}
    \right\}
   \right\|_{(2,2)}\geq\tau\right\} \nonumber\\
   \leq& 2k_TL\exp\left\{-\frac{3}{64}\frac{m}{k_T\mu_{\max}\kappa_{\max}}\frac{\sqrt{\log n}}{\log n(3\sqrt{\log n}+1)}\right\}\nonumber\\
   \leq& 2k_TL\exp\left\{-\frac{3}{256}\frac{m}{k_T\mu_{\max}\kappa_{\max}}\frac{1}{\log n}\right\}\nonumber\\
   \leq& 2k_TL\exp\left\{-3\log n\right\}\leq2n^{-2}.
\end{align*}
Now taking a union bound over all $j=1,\cdots,l$ gives
\begin{equation}
Pr\left\{\textrm{\eqref{batch_inequalities} holds for all
$j=1,\cdots,l$}\right\}\geq1-2n^{-2}l \geq1-2n^{-2}(\log
n+1)\geq1-2n^{-1},\nonumber
\end{equation}
which finishes the proof.
\end{proof}

From Lemma \ref{batch_shrink_lemma}, the following chains of
contractions hold with probability at least $1-2n^{-1}$:
\begin{align}
    \|\mathbf{Q}_{(1)}\|_{F}\leq&\frac{1}{2\sqrt{\log n}}\|\mathbf{Q}_{(0)}\|_{F}, \label{Q_contraction_1}\\
    \|\mathbf{Q}_{(2)}\|_{F}\leq&\frac{1}{4\log n}\|\mathbf{Q}_{(0)}\|_{F}, \nonumber\\
    &\vdots \nonumber\\
    \|\mathbf{Q}_{(l)}\|_{F}\leq&\prod_{j=1}^{l}c_{j}\|\mathbf{Q}_{(0)}\|_{F}\leq \frac{1}{\log n}\frac{1}{2^{l}}\|\mathbf{Q}_{(0)}\|_{F}. \label{Q_contraction_2}
\end{align}

Finally, we set $\mathbf{W}$ so that
\begin{equation}\label{definition_dual_W}
    \mathcal{P}_{\Omega^{c}}\mathbf{W}
    =\left[
           \sum_{j=1}^{l}\frac{m}{m_{j}}\mathcal{P}_{K_{1j}}\mathbf{A}_{(1)}\mathcal{P}_{T}\mathbf{q}_{(j-1)1}, \cdots, \sum_{j=1}^{l}\frac{m}{m_{j}}\mathcal{P}_{K_{Lj}}\mathbf{A}_{(L)}\mathcal{P}_{T}\mathbf{q}_{(j-1)L}
    \right],
\end{equation}
and $\mathcal{P}_{\Omega}\mathbf{W}=\mathbf{0}$. Also, set
$\mathbf{U}$ to be
\begin{align}\label{definition_dual_U}
    \mathbf{U}=&\left[\mathbf{\Sigma}_{(1)}^{-1}\mathbf{A}_{(1)}'\mathcal{P}_{\Omega_{1}^{c}}\mathbf{w}_{1}
     ,\cdots,\mathbf{\Sigma}_{(L)}^{-1}\mathbf{A}_{(L)}'\mathcal{P}_{\Omega_{L}^{c}}\mathbf{w}_{L}\right] \nonumber\\
              =&\left[\sum_{j=1}^{l}\frac{m}{m_{j}}\mathbf{\Sigma}_{(1)}^{-1}\mathbf{A}_{(1)}'
              \mathcal{P}_{K_{1j}}\mathbf{A}_{(1)}\mathcal{P}_{T}\mathbf{q}_{(j-1)1},\cdots
              ,\sum_{j=1}^{l}\frac{m}{m_{j}}\mathbf{\Sigma}_{(L)}^{-1}\mathbf{A}_{(L)}'
              \mathcal{P}_{K_{Lj}}\mathbf{A}_{(L)}\mathcal{P}_{T}\mathbf{q}_{(j-1)L}\right] \nonumber\\
              =&\sum_{j=1}^{l}\frac{m}{m_{j}}\left[\tilde{\mathbf{A}}_{(1,j)}\mathcal{P}_{T}\mathbf{q}_{(j-1)1}
              ,\cdots,\tilde{\mathbf{A}}_{(L,j)}\mathcal{P}_{T}\mathbf{q}_{(j-1)L}\right].
\end{align}
Having set all of these, we are now ready to prove Theorem
\ref{theorem_dual_certificate}. The proof is given in Appendix
\ref{dual_construct} .

\section{conclusion}
\label{sec:conclusion}

This paper proposes the robust group lasso (RGL) model that
recovers a group sparse signal matrix for sparsely corrupted
measurements. The RGL model minimizes the mixed
$\ell_{2,1}$/$\ell_1$-norm under linear measurement constraints,
and hence is convex. We establish the recoverability of the RGL
model, showing that the true group sparse signal matrix and the
sparse error matrix can be exactly recovered with a high probability
under certain conditions. Our theoretical analysis provides a
solid performance guarantee to the RGL model.

\appendices
\section{Proof of Lemma \ref{lemma_BEQ_1}}\label{proof_3}
\begin{IEEEproof}
Here we prove the second part of Lemma \ref{lemma_BEQ_1}. By
definitions
$\tilde{\mathbf{A}}_{(i)}=\mathbf{\Sigma}_{(i)}^{-1}\mathbf{A}_{(i)}'\mathcal{P}_{\Omega_{i}^*}\mathbf{A}_{(i)}$,
it holds
\begin{align*}
&\mathcal{P}_{T}\left(\frac{m}{m-k_{\max}}\tilde{\mathbf{A}}_{(i)}\mathbf{\Sigma}_{(i)}^{-1}-\mathbf{\Sigma}_{(i)}^{-1}\right)
    \mathcal{P}_{T}\\
=&\mathcal{P}_{T}\left(\frac{m}{m-k_{\max}}\mathbf{\Sigma}_{(i)}^{-1}\mathbf{A}_{(i)}'\mathcal{P}_{\Omega_{i}^*}\mathbf{A}_{(i)}\mathbf{\Sigma}_{(i)}^{-1}
-\mathbf{\Sigma}_{(i)}^{-1}\right)\mathcal{P}_{T}\\
=&\sum_{j\in\Omega_i^*}\mathbf{M}_{(j)},
\end{align*}
where
\[\mathbf{M}_{(j)}\triangleq\frac{1}{m-k_{\max}}\mathcal{P}_{T}\left(\mathbf{\Sigma}_{(i)}^{-1}\mathbf{a}_{(i)j}\mathbf{a}_{(i)j}'\mathbf{\Sigma}_{(i)}^{-1}
-\mathbf{\Sigma}_{(i)}^{-1}\right)\mathcal{P}_{T}.\] Since
$\expect{\mathbf{a}_{(i)j}\mathbf{a}_{(i)j}'}=\mathbf{\Sigma}_{(i)}$,
it is obvious that that $\expect{\mathbf{M}_{(j)}}=0$. We estimate
the induced $\ell_{(2,2)}$-norm of $\mathbf{M}_{(j)}$ in order to
implement the matrix Berstein inequality later. It holds
\begin{align}
\|\mathbf{M}_{(j)}\|_{(2,2)}
=&\left\|\frac{1}{m-k_{\max}}\mathcal{P}_{T}\left(\mathbf{\Sigma}_{(i)}^{-1}\mathbf{a}_{(i)j}\mathbf{a}_{(i)j}'\mathbf{\Sigma}_{(i)}^{-1}
-\mathbf{\Sigma}_{(i)}^{-1}\right)\mathcal{P}_{T}\right\|_{(2,2)}\nonumber\\
\leq&\frac{1}{m-k_{\max}}\left(\left\|\mathcal{P}_{T}\left(\mathbf{\Sigma}_{(i)}^{-1}\mathbf{a}_{(i)j}\mathbf{a}_{(i)j}'
\mathbf{\Sigma}_{(i)}^{-1}\right)\mathcal{P}_{T}\right\|_{(2,2)}
+\left\|\mathcal{P}_T\mathbf{\Sigma}_{(i)}^{-1}\mathcal{P}_{T}\right\|_{(2,2)}\right)\nonumber\\
\leq&\frac{1}{m-k_{\max}}\left(\left\|\mathcal{P}_{T}\left(\mathbf{\Sigma}_{(i)}^{-1}\mathbf{a}_{(i)j}\mathbf{a}_{(i)j}'
\mathbf{\Sigma}_{(i)}^{-1}\right)\mathcal{P}_{T}\right\|_{(2,2)}
+\kappa_{i}\right)\nonumber\\
=&\frac{1}{m-k_{\max}}\left(\left\|\mathcal{P}_{T}\mathbf{\Sigma}_{(i)}^{-1}\mathbf{a}_{(i)j}\right\|_2^2+\kappa_i\right)\leq
\frac{1}{m-k_{\max}}(\mu_ik_T+\kappa_{i}), \nonumber
\end{align}
where the first inequality follows from the triangle inequality
and the last inequality follows from Assumption
\eqref{incoherence2}. Since $\kappa_i\geq1$ and $\mu_i\geq1$, the
above bound on $\|\mathbf{M}_{(j)}\|_{(2,2)}$ can be further
relaxed as
\[\|\mathbf{M}_{(j)}\|_{(2,2)}\leq\frac{2\kappa_i\mu_ik_T}{m-k_{\max}}\triangleq B.\]

Meanwhile, since
$\mathbf{M}_{(j)}'\mathbf{M}_{(j)}=\mathbf{M}_{(j)}\mathbf{M}_{(j)}'$,
we only need to consider one of them.
\begin{align}
&\left\|\expect{\mathbf{M}_{(j)}'\mathbf{M}_{(j)}}\right\|_{(2,2)}\nonumber\\
=&\frac{1}{(m-k_{\max})^2}\left\|\expect{\mathcal{P}_{T}\mathbf{\Sigma}_{(i)}^{-1}\mathbf{a}_{(i)j}\left(\mathbf{a}_{(i)j}'\mathbf{\Sigma}_{(i)}^{-1}\mathcal{P}_{T}
\mathbf{\Sigma}_{(i)}^{-1}\mathbf{a}_{(i)j}\right)\mathbf{a}_{(i)j}'\mathbf{\Sigma}_{(i)}^{-1}\mathcal{P}_{T}
-\left(\mathcal{P}_{T}\mathbf{\Sigma}_{(i)}^{-1}\mathcal{P}_{T}\right)^2}\right\|_{(2,2)}\nonumber\\
=&\frac{1}{(m-k_{\max})^2}
\left\|\expect{\left\|\mathcal{P}_{T}\mathbf{\Sigma}_{(i)}^{-1}\mathbf{a}_{(i)j}\right\|_2^2
\mathcal{P}_{T}\mathbf{\Sigma}_{(i)}^{-1}\mathbf{a}_{(i,j)}\mathbf{a}_{(i)j}'\mathbf{\Sigma}_{(i)}^{-1}\mathcal{P}_{T}}
-\left(\mathcal{P}_{T}\mathbf{\Sigma}_{(i)}^{-1}\mathcal{P}_{T}\right)^2\right\|_{(2,2)}\nonumber\\
\leq&\frac{1}{(m-k_{\max})^2}\left(\left\|\expect{\left\|\mathcal{P}_{T}\mathbf{\Sigma}_{(i)}^{-1}\mathbf{a}_{(i)j}\right\|_2^2
\mathcal{P}_{T}\mathbf{\Sigma}_{(i)}^{-1}\mathbf{a}_{(i,j)}\mathbf{a}_{(i)j}'\mathbf{\Sigma}_{(i)}^{-1}\mathcal{P}_{T}}\right\|_{(2,2)}+\kappa_i^2\right)
\nonumber\\
\leq&\frac{1}{(m-k_{\max})^2}\left(\mu_ik_T
\left\|\expect{\mathcal{P}_{T}\mathbf{\Sigma}_{(i)}^{-1}\mathbf{a}_{(i)j}\mathbf{a}_{(i)j}'\mathbf{\Sigma}_{(i)}^{-1}\mathcal{P}_{T}}\right\|_{(2,2)}
+\kappa_i^2\right)\nonumber\\
\leq&\frac{\kappa_i\mu_ik_T+\kappa_i^2}{(m-k_{\max})^2}\leq\frac{\kappa_i^2(\mu_ik_T+1)}{(m-k_{\max})^2}
\leq\frac{2\kappa_i^2\mu_ik_T}{(m-k_{\max})^2},\nonumber
\end{align}
where the first equality follows from straight-up calculation using $\expect{\mathbf{a}_{(i)j}\mathbf{a}_{(i)j}'}=\mathbf{\Sigma}_{(i)}$. The first inequality follows from triangle inequality, the second inequality follows from the definition of incoherence \eqref{incoherence2}, and the rest of the inequalities uses the fact that $\kappa_i\geq1$ and $\mu_i\geq1$.
Thus, by triangle inequality,
\[\left\|\expect{\sum_{j\in\Omega_i^*}\mathbf{M}_{(j)}'\mathbf{M}_{(j)}}\right\|_{(2,2)}\leq\frac{2\kappa_i^2\mu_ik_T}{(m-k_{\max})^2}\cdot (m-k_{\max})
=\frac{2\kappa_i^2\mu_ik_T}{m-k_{\max}}\triangleq\sigma^2.\]
Plugging $B$ and $\sigma^2$ into the Matrix Berstein inequality, we finish the proof of Lemma \ref{lemma_BEQ_1}.
\end{IEEEproof}

\section{Proof of Lemma \ref{lemma_BEQ_2}} \label{proof_4}
\begin{IEEEproof}
We use the vector Berstein inequality to prove the lemma. Picking
any $k\in T^c$ and any $i\in\{1,\cdots,L\}$, we have
\[\tilde{\mathbf{A}}_{(i)}\mathbf{e}_k=\frac{1}{m}\sum_{j\in\Omega_i^*}\langle\mathbf{a}_{(i)j},\mathbf{e}_{k}\rangle
\mathbf{\Sigma}_{(i)}^{-1}\mathbf{a}_{(i)j},\] Letting
\[\mathbf{g}_{(i,j)}=\frac{1}{m}\langle\mathbf{a}_{(i)j},\mathbf{e}_{k}\rangle\mathcal{P}_{T}\mathbf{\Sigma}_{(i)}^{-1}\mathbf{a}_{(i)j},\]
then it holds
\begin{equation}\label{tttttttt1}
    \mathcal{P}_{T}\tilde{\mathbf{A}}_{(i)}\mathbf{e}_{k}
    =\sum_{j\in\Omega_i^*}\mathbf{g}_{(i,j)}.
\end{equation}

Since $\{\mathbf{a}_{(i)j}\}_{j\in\Omega_i^*}$ are i.i.d. samples
from $\mathcal{F}_i$, the sequence of vectors
$\left\{\mathbf{g}_{(i,j)}\right\}_{j\in\Omega_i^*}$ are i.i.d.
random variables. In order to apply the vector Berstein
inequality, we first need to show that
$\expect{\mathbf{g}_{(i,j)}}=0$ for any $j\in\Omega_i^*$:
 \begin{equation}
    \expect{\mathbf{g}_{(i,j)}}
    =\frac{1}{m}\expect{\langle\mathbf{a}_{(i)j},\mathbf{e}_{k}\rangle\mathcal{P}_{T}\mathbf{\Sigma}_{(i)}^{-1}\mathbf{a}_{(i)j}}
    =\frac{1}{m}\mathcal{P}_{T}\mathbf{\Sigma}_{(i)}^{-1}\expect{\mathbf{a}_{(i)j}\mathbf{a}_{(i)j}'}\mathbf{e}_{k}
    =\frac{1}{m}\mathcal{P}_{T}\mathbf{e}_{k}=0. \nonumber
 \end{equation}
The last equality is true since $k\in T^{c}$. Second, we
calculate the bound $B$ for any single
$\left\|\mathbf{g}_{(i,j)}\right\|_{2}$:
 \begin{align}
     \|\mathbf{g}_{(i,j)}\|_{2}^2=&\frac{1}{m^2}
     |\langle\mathbf{a}_{(i)j},\mathbf{e}_{k}\rangle|^2\|\mathcal{P}_{T}\mathbf{\Sigma}_{(i)}^{-1}\mathbf{a}_{(i)j}\|_{2}^2
     \leq\frac{\mu_{i}\|\mathcal{P}_{T}\mathbf{\Sigma}_{(i)}^{-1}\mathbf{a}_{(i)j}\|_{2}^2}{m^2}\leq\frac{\mu_{i}^2k_{T}}{m^2},\nonumber
 \end{align}
 where the first inequality follows from the incoherence condition \eqref{incoherence1} and the second inequality follows from \eqref{incoherence2}.
 Furthermore, we have
 \begin{align}
    \expect{\left\|\mathbf{g}_{(i,j)}\right\|_{2}^{2}}=&\frac{1}{m^{2}}
    \expect{\left(\langle\mathbf{a}_{(i)j},\mathbf{e}_{k}\rangle\mathcal{P}_{T}\mathbf{\Sigma}_{(i)}^{-1}\mathbf{a}_{(i)j}\right)'
    \left(\langle\mathbf{a}_{(i)j},\mathbf{e}_{k}\rangle\mathcal{P}_{T}\mathbf{\Sigma}_{(i)}^{-1}\mathbf{a}_{(i)j}\right)} \nonumber\\
    \leq&\frac{1}{m^{2}}\mu_{i} \expect{\mathbf{a}_{(i)j}'\mathbf{\Sigma}_{(i)}^{-1}\mathcal{P}_{T}\mathbf{\Sigma}_{(i)}^{-1}\mathbf{a}_{(i)j}}\nonumber\\
    =&\frac{1}{m^{2}}\mu_{i}\cdot \textrm{Tr}\left(\expect{\mathbf{a}_{(i)j}\mathbf{a}_{(i)j}'}\mathbf{\Sigma}_{(i)}^{-1}\mathcal{P}_{T}\mathbf{\Sigma}_{(i)}^{-1}\right)\nonumber\\
    =&\frac{1}{m^{2}}\mu_{i}\cdot \textrm{Tr}\left(\mathcal{P}_{T}\mathbf{\Sigma}_{(i)}^{-1}\right)
    \leq\frac{\mu_{i} k_{T}\kappa_{i}}{m^2}, \nonumber
 \end{align}
where $\textrm{Tr}(\cdot)$ denotes the trace of a matrix. The
first inequality follows from the incoherence property
\eqref{incoherence1}. The last inequality follows from the fact
that $\mathcal{P}_{T}\mathbf{\Sigma}_{(i)}^{-1}$ is of rank at
most $k_T$ so that its trace is upper bounded by $k_T\kappa_i$.
 Thus, it holds
 \begin{equation}
    \sum_{j\in\Omega_i^*}\expect{\left\|\mathbf{g}_{(i,j)}\right\|_{2}^{2}}
    \leq\sum_{j\in\Omega_i^*}\frac{\mu_{\max} \kappa_{\max}k_{T}}{m^2}\leq\frac{\mu_{\max}\kappa_{\max}
    k_{T}}{m}\triangleq\sigma^2.
 \end{equation}
 Substituting the above bound to the vector Berstein inequality
 yields
\begin{equation}
Pr\left(\left\|\sum_{j\in\Omega_i^*}\mathbf{g}_{(i,j)}\right\|_2\geq
t\right)\leq
\textrm{exp}\left(-\frac{t^2}{\frac{8\mu_{\max}\kappa_{\max}
k_{T}}{m}}+\frac{1}{4}\right), \nonumber
\end{equation}
given $\sigma^{2}/B=\sqrt{k_{1}}\kappa_{max}\geq1$. Let
$t=\sqrt{C\log n\frac{\mu_{\max}\kappa_{\max} k_{T}}{m}}$. Using
the fact that $k_{T}\leq\alpha\frac{m}{\mu_{\max}\kappa_{\max}\log
n}$, it holds $t\leq\sqrt{C\alpha}$ when $\alpha\leq\frac{1}{24}$.
we can choose $C=24$ such that $C\alpha\leq1$, which guarantees
$t\leq 1$ and gives
\begin{equation}
Pr\left(\left\|\sum_{j\in\Omega_i^*}\mathbf{g}_{(i,j)}\right\|_2\geq
1\right)\leq e^{\frac{1}{4}}n^{-3}. \nonumber
\end{equation}

Recalling \eqref{tttttttt1} and taking a union bound over all
$k\in T^c$ and $i\in\{1,\cdots,L\}$, we have
\begin{align*}
&Pr\left(\max_{i\in\{1,\cdots,L\},k\in T^{c}}
    \left\|\mathcal{P}_{T}\tilde{\mathbf{A}}_{(i)}\mathbf{e}_{k}\right\|_2\geq1\right)\\
\leq&\sum_{i=1}^L\sum_{k\in T^c}Pr\left(
    \left\|\sum_{j\in\Omega_i^*}\mathbf{g}_{(i,j)}\right\|_{2}\geq1\right)\\
\leq& k_TLe^{\frac{1}{4}}n^{-3}\leq e^{\frac{1}{4}}n^{-2},
\end{align*}
where the last inequality follows from the fact $k_TL\leq n$. This
completes the proof.
\end{IEEEproof}

\section{Proof of Lemma \ref{equivalence_relation}}\label{appendix_equivalence_relation}
\begin{proof}
Since $\overline{\mathbf{Y}}$ and $\overline{\mathbf{S}}$ are the
true group sparse signal and sparse error matrices, respectively,
they satisfy the measurement equation
$$\mathbf{M} = [\mathbf{A}_{(1)}\bar{\mathbf{y}}_{1},
\cdots, \mathbf{A}_{(L)}\bar{\mathbf{y}}_{L}] +
\overline{\mathbf{S}}.$$ Furthermore, since
$(\overline{\mathbf{Y}}+\mathbf{H},\overline{\mathbf{S}}-\mathbf{F})$
is an optimal solution to the RGL model (\ref{ee3}), they must
also satisfy the constraint
$$\mathbf{M} = [\mathbf{A}_{(1)}(\bar{\mathbf{y}}_{1}+\mathbf{h}_{1}),
\cdots, \mathbf{A}_{(L)}(\bar{\mathbf{y}}_{L}+\mathbf{h}_{L})] +
\overline{\mathbf{S}}-\mathbf{F}.$$ Subtracting these two
equations yields result i) of Lemma \ref{equivalence_relation}.

Since the objective function of (\ref{ee3}) is convex, we
obtain an inequality
\begin{align}\label{e52-1}
\|\overline{\mathbf{Y}}+\mathbf{H}\|_{2,1}+\lambda\|\overline{\mathbf{S}}-\mathbf{F}\|_{1}
&\geq\|\overline{\mathbf{Y}}\|_{2,1}+\lambda\|\overline{\mathbf{S}}\|_{1}
+\langle\partial\|\overline{\mathbf{Y}}\|_{2,1},\mathbf{H}\rangle-\lambda\langle\partial\|\overline{\mathbf{S}}\|_{1},\mathbf{F}\rangle,
\end{align}
where $\partial\|\overline{\mathbf{Y}}\|_{2,1}$ denotes a subgradient of
the $\ell_{2,1}$-norm at $\overline{\mathbf{Y}}$ and
$\partial\|\overline{\mathbf{S}}\|_{1}$ denotes a subgradient of the
$\ell_1$-norm at $\overline{\mathbf{S}}$. Furthermore, the corresponding subgradients can be written
as
\begin{align*}
\partial\|\overline{\mathbf{Y}}\|_{2,1} &= \overline{\mathbf{V}}+\overline{\mathbf{R}},\\
\partial\|\overline{\mathbf{S}}\|_{1}   &= \textrm{sgn}(\overline{\mathbf{S}})+\overline{\mathbf{Q}},
\end{align*}
where $\overline{\mathbf{V}} \in \mathcal{R}^{n \times L}$
satisfies
$(\mathcal{P}_{T}\overline{\mathbf{V}})^i=\frac{\overline{\mathbf{y}}^i}{\|\overline{\mathbf{y}}^i\|_{2}}$
and $(\mathcal{P}_{T^{c}}\overline{\mathbf{V}})^{i}=\mathbf{0}$,
$\forall i=1,\cdots,n$; $\overline{\mathbf{R}} \in \mathcal{R}^{n
\times L}$ satisfies
$\mathcal{P}_{T}\overline{\mathbf{R}}=\mathbf{0}$ and
$\|\mathcal{P}_{T^{c}}\overline{\mathbf{R}}\|_{2,\infty} \leq 1$;
$\overline{\mathbf{Q}} \in \mathcal{R}^{m \times L}$ satisfies
$\mathcal{P}_{\Omega}\overline{\mathbf{Q}}=\mathbf{0}$ and
$\|\mathcal{P}_{\Omega^{c}}\overline{\mathbf{Q}}\|_{\infty} \leq
1$. Therefore, we have
\begin{align}\label{e52}
\|\overline{\mathbf{Y}}+\mathbf{H}\|_{2,1}+\lambda\|\overline{\mathbf{S}}-\mathbf{F}\|_{1}
\geq \|\overline{\mathbf{Y}}\|_{2,1}+\lambda\|\overline{\mathbf{S}}\|_{1}
+\langle\overline{\mathbf{V}}+\overline{\mathbf{R}},\mathbf{H}\rangle-\lambda\langle
sgn(\overline{\mathbf{S}})+\overline{\mathbf{Q}},\mathbf{F}\rangle,
\end{align}
for any $\overline{\mathbf{R}}$ and $\overline{\mathbf{Q}}$
satisfying the conditions mentioned above.

We construct a specific pair of $\overline{\mathbf{R}}$ and
$\overline{\mathbf{Q}}$ in the following way. Let
\[\bar{\mathbf{r}}^{i}=
\left\{
  \begin{array}{ll}
    \frac{\mathbf{h}^{i}}{\|\mathbf{h}^{i}\|_{2}}, & \hbox{if $\mathbf{h}^{i}\neq \mathbf{0}'$ and $i \in T^{c}$;} \\
    \mathbf{0}', & \hbox{otherwise.}
  \end{array}
\right.
\]
where $\mathbf{h}^{i}$ and $\bar{\mathbf{r}}^i$ are the $i$-th row
of $\mathbf{H}$ and $\overline{\mathbf{R}}$, respectively.
Meanwhile, let
$\overline{\mathbf{Q}}=-\textrm{sgn}(\mathcal{P}_{\Omega^{c}}\mathbf{F})$.
It follows that
\begin{align*}
    \langle\overline{\mathbf{R}},\mathbf{H}\rangle=\|\mathcal{P}_{T^{c}}\mathbf{H}\|_{2,1}, \\
    \langle\overline{\mathbf{Q}},\mathbf{F}\rangle=-\|\mathcal{P}_{\Omega^{c}}\mathbf{F}\|_{1}.
\end{align*}
Substituting the above equalities into (\ref{e52}) gives result
ii) of Lemma \ref{equivalence_relation}.
\end{proof}

\section{Proof of lemma \ref{null_space_lemma}}\label{appendix_null_space}
\begin{proof}
We first show that $\mathcal{P}_{T} \mathbf{H}=\mathbf{0}$. Since
$\left[\mathbf{A}_{(1)}\mathbf{h}_{1},
\cdots,\mathbf{A}_{(L)}\mathbf{h}_{L}\right]=\mathbf{F}$ and
$\mathcal{P}_{\Omega^c}\mathbf{F}=\mathbf{0}$, it holds
\[[\mathcal{P}_{\Omega_1^{c}}\mathbf{A}_{(1)}\mathbf{h}_{1},\cdots,\mathcal{P}_{\Omega_L^{c}}\mathbf{A}_{(L)}\mathbf{h}_{L}]=0.\]
Meanwhile, $\mathcal{P}_{T^c} \mathbf{H} = \mathbf{0}$ implies
$[\mathcal{P}_{\Omega_1^{c}}\mathbf{A}_{(1)}\mathcal{P}_{T^c}\mathbf{h}_{1},\cdots,\mathcal{P}_{\Omega_L^{c}}\mathbf{A}_{(L)}\mathcal{P}_{T^c}\mathbf{h}_{L}]=0$.
Therefore, it holds
\begin{align*}
&[\mathcal{P}_{\Omega_1^{c}}\mathbf{A}_{(1)}\mathcal{P}_{T}\mathbf{h}_{1},\cdots,\mathcal{P}_{\Omega_L^{c}}\mathbf{A}_{(L)}\mathcal{P}_{T}\mathbf{h}_{L}]\\
=&[\mathcal{P}_{\Omega_1^{c}}\mathbf{A}_{(1)}\mathbf{h}_{1},\cdots,\mathcal{P}_{\Omega_L^{c}}\mathbf{A}_{(L)}\mathbf{h}_{L}]
-[\mathcal{P}_{\Omega_1^{c}}\mathbf{A}_{(1)}\mathcal{P}_{T^c}\mathbf{h}_{1},\cdots,\mathcal{P}_{\Omega_L^{c}}\mathbf{A}_{(L)}\mathcal{P}_{T^c}\mathbf{h}_{L}]
=\mathbf{0}.
\nonumber
\end{align*}

Since for any $i=1,\cdots,L$, $\Omega_i^*$ is a subset of
$\Omega_i^c$, it follows
\[\left[\mathcal{P}_{\Omega_1^*}\mathbf{A}_{(1)}\mathcal{P}_{T}\mathbf{h}_{1},\cdots,\mathcal{P}_{\Omega_L^*}\mathbf{A}_{(L)}\mathcal{P}_{T}\mathbf{h}_{L}\right]=\mathbf{0},\]
and consequently
\begin{align*}
&\textbf{Blkdiag}\left\{\mathcal{P}_T\frac{m}{m-k_{\max}}\tilde{\mathbf{A}}_{(1)}\mathcal{P}_T
,\cdots,\mathcal{P}_T\frac{m}{m-k_{\max}}\tilde{\mathbf{A}}_{(L)}\mathcal{P}_T\right\}\cdot\textrm{vec}(\mathbf{H})\\
=&\frac{m}{m-k_{\max}}\textrm{vec}\left(\left[\mathcal{P}_T\mathbf{\Sigma}_{(1)}^{-1}\mathbf{A}_{(1)}'\mathcal{P}_{\Omega_1^*}\mathbf{A}_{(1)}\mathcal{P}_{T}\mathbf{h}_{1}
,\cdots,\mathcal{P}_T\mathbf{\Sigma}_{(L)}^{-1}\mathbf{A}_{(L)}'\mathcal{P}_{\Omega_L^*}\mathbf{A}_{(L)}\mathcal{P}_{T}\mathbf{h}_{L}\right]\right)\\
=&\mathbf{0}.
\end{align*}
This equality implies
\begin{align*}
\left\|\textbf{Blkdiag}\left\{\mathcal{P}_T\left(\frac{m}{m-k_{\max}}\tilde{\mathbf{A}}_{(1)}-\mathbf{I}\right)\mathcal{P}_T
,\cdots,\mathcal{P}_T\left(\frac{m}{m-k_{\max}}\tilde{\mathbf{A}}_{(L)}-\mathbf{I}\right)\mathcal{P}_T\right\}\cdot\textrm{vec}(\mathbf{H})\right\|_2
=\|\mathcal{P}_T\mathbf{H}\|_F.
\end{align*}

On the other hand, according to \eqref{BEQ_corollary_plus}, it
follows with a high probability
\begin{align*}
&\left\|\textbf{Blkdiag}\left\{\mathcal{P}_T\left(\frac{m}{m-k_{\max}}\tilde{\mathbf{A}}_{(1)}-\mathbf{I}\right)\mathcal{P}_T,\cdots,\mathcal{P}_T\left(\frac{m}{m-k_{\max}}\tilde{\mathbf{A}}_{(L)}-\mathbf{I}\right)\mathcal{P}_T\right\}\cdot\textrm{vec}(\mathbf{H})\right\|_2\\
\leq&\left\|\textbf{Blkdiag}\left\{\mathcal{P}_T\left(\frac{m}{m-k_{\max}}\tilde{\mathbf{A}}_{(1)}-\mathbf{I}\right)\mathcal{P}_T,\cdots,\mathcal{P}_T\left(\frac{m}{m-k_{\max}}\tilde{\mathbf{A}}_{(L)}-\mathbf{I}\right)\mathcal{P}_T\right\}\right\|_{(2,2)}\cdot
\|\mathcal{P}_T\mathbf{H}\|_F\\
\leq&\frac{1}{2\sqrt{\log n}}\|\mathcal{P}_T\mathbf{H}\|_F.
\end{align*}
Thus,
\[\|\mathcal{P}_T\mathbf{H}\|_F\leq\frac{1}{2\sqrt{\log n}}\|\mathcal{P}_T\mathbf{H}\|_F,\]
which implies $\mathcal{P}_T\mathbf{H}=\mathbf{0}$. Because
$\mathcal{P}_{T^c}\mathbf{H}=\mathbf{0}$, we have
$\mathbf{H}=\mathbf{0}$. Since
$\mathbf{F}=\left[\mathbf{A}_{(1)}\mathbf{h}_{1},\cdots,\mathbf{A}_{(L)}\mathbf{h}_{L}\right]$,
it follows that $\mathbf{F}=0$.
\end{proof}

\section{Finishing the Proof of Theorem \ref{theorem_inexact_duality}: Inexact Duality}\label{appendix_inexact_duality}

This section is dedicated to proving that with a high probability
\eqref{interim_inexact_3} implies
$\mathcal{P}_{T^c}\mathbf{H}=\mathbf{0}$ and
$\mathcal{P}_{\Omega^c}\mathbf{F}=\mathbf{0}$. To do so, we first
derive an upper bound for $\|\mathcal{P}_{T} \mathbf{H}\|_F$,
expressed as a linear combination of
$\|\mathcal{P}_{T^{c}}\mathbf{H}\|_{2,1}$ and
$\|\mathcal{P}_{\Omega^{c}}\mathbf{F}\|_1$.

Using \eqref{BEQ_corollary}, it follows
\begin{align}
    & \left\|
    \textbf{BLKdiag}\left\{\mathcal{P}_{T}\left(\frac{m}{m-k_{\max}}\tilde{\mathbf{A}}_{(1)}-\mathbf{I}\right)\mathcal{P}_{T},\cdots,
    \mathcal{P}_{T}\left(\frac{m}{m-k_{\max}}\tilde{\mathbf{A}}_{(L)}-\mathbf{I}\right)\mathcal{P}_{T}\right\}
    \textrm{vec}(\mathbf{H})
    \right\|_2 \leq \frac{1}{2} \|\mathcal{P}_{T}\mathbf{H}\|_F. \nonumber
\end{align}
Since $\left\|
    \textbf{BLKdiag}\left\{\mathcal{P}_{T},\cdots,
    \mathcal{P}_{T}\right\}
    \textrm{vec}(\mathbf{H})
    \right\|_2 = \|\mathcal{P}_{T}\mathbf{H}\|_F$, applying
    the triangle inequality yields
\begin{equation}
      \|\mathcal{P}_{T}\mathbf{H}\|_F\leq
     2\left\|\frac{m}{m-k_{\max}}
    \textbf{BLKdiag}\left\{\mathcal{P}_{T}\tilde{\mathbf{A}}_{(1)}\mathcal{P}_{T},\cdots,
    \mathcal{P}_{T}\tilde{\mathbf{A}}_{(L)}\mathcal{P}_{T}\right\}
    \textrm{vec}(\mathbf{H})\right\|_2.   \nonumber
\end{equation}
Observing
$\textrm{vec}(\mathcal{P}_{T}\mathbf{H})=\textrm{vec}(\mathbf{H})-\textrm{vec}(\mathcal{P}_{T^{c}}\mathbf{H})$
and using the triangle inequality again, we have
\begin{align}\label{interim_appendix_1}
     \|\mathcal{P}_{T}\mathbf{H}\|_{F} \leq
     & 2 \left\|\frac{m}{m-k_{\max}}\textbf{BLKdiag}\left\{\mathcal{P}_{T}\tilde{\mathbf{A}}_{(1)},
     \cdots,\mathcal{P}_{T}\tilde{\mathbf{A}}_{(L)}
     \right\}\textrm{vec}(\mathbf{H})\right\|_{2} \nonumber \\
     & + 2 \left\|\frac{m}{m-k_{\max}}\textbf{BLKdiag}\left\{\mathcal{P}_{T}\tilde{\mathbf{A}}_{(1)},
     \cdots,\mathcal{P}_{T}\tilde{\mathbf{A}}_{(L)}
     \right\}\textrm{vec}(\mathcal{P}_{T^c}\mathbf{H})\right\|_{2}.
\end{align}

Below, we upper bound the two terms at the right-hand side of
\eqref{interim_appendix_1}, respectively.

\noindent \textcircled{1} \textbf{Bounding the First Term of \eqref{interim_appendix_1}:}
By definitions
$\tilde{\mathbf{A}}_{(i)}=\mathbf{\Sigma}_{(i)}^{-1}\mathbf{A}_{(i)}'\mathcal{P}_{\Omega_{i}^*}\mathbf{A}_{(i)}$,
it follows
\begin{align}\label{interim_appendix_2}
     &  \left\|\frac{m}{m-k_{\max}}\textbf{BLKdiag}\left\{\mathcal{P}_{T}\tilde{\mathbf{A}}_{(1)},
     \cdots,\mathcal{P}_{T}\tilde{\mathbf{A}}_{(L)}
     \right\}\textrm{vec}(\mathbf{H})\right\|_{2} \nonumber \\
    = & \left\|\frac{m}{m-k_{\max}}
     \textrm{vec}\left(\left[\mathcal{P}_{T}\mathbf{\Sigma}_{(1)}^{-1}\mathbf{A}_{(1)}'\mathcal{P}_{\Omega_{1}^*}\mathbf{A}_{(1)}\mathbf{h}_1,
     \cdots,\mathcal{P}_{T}\mathbf{\Sigma}_{(L)}^{-1}\mathbf{A}_{(L)}'\mathcal{P}_{\Omega_{L}^*}\mathbf{A}_{(L)}\mathbf{h}_L\right]
     \right)\right\|_{2}\nonumber\\
    = & \left\|\frac{m}{m-k_{\max}}\textrm{vec}\left(\left[\mathcal{P}_{T}\mathbf{\Sigma}_{(1)}^{-1}\mathbf{A}_{(1)}'\mathcal{P}_{\Omega_{1}^*}\mathbf{f}_1,
     \cdots,\mathcal{P}_{T}\mathbf{\Sigma}_{(L)}^{-1}\mathbf{A}_{(L)}'\mathcal{P}_{\Omega_{L}^*}\mathbf{f}_L\right]
     \right)\right\|_{2}.
\end{align}
Here the second equality comes from result i) in Lemma
\ref{equivalence_relation}, namely,
$\mathbf{F}=[\mathbf{A}_{(1)}\mathbf{h}_1,\cdots,\mathbf{A}_{(L)}\mathbf{h}_L]$.
Recalling that $\Omega_i^*$ is a subset of $\Omega_i^c$ for any
$i=1,\cdots,L$, we have
\[\left\|\textrm{vec}\left(\left[\mathcal{P}_{\Omega_{1}^*}\mathbf{f}_1,\cdots,\mathcal{P}_{\Omega_{L}^*}\mathbf{f}_L\right]\right)\right\|_2
\leq\|\textrm{vec}(\mathcal{P}_{\Omega^{c}}\mathbf{F})\|_2.\]
Based on this inequality, we upper bound
\eqref{interim_appendix_2} using the induced norm property
\begin{align}
     & \left\|\frac{m}{m-k_{\max}}\textbf{BLKdiag}\left\{\mathcal{P}_{T}\tilde{\mathbf{A}}_{(1)},
      \cdots,\mathcal{P}_{T}\tilde{\mathbf{A}}_{(L)}
      \right\}\textrm{vec}(\mathbf{H})\right\|_{2}\nonumber\\
\leq & \left\|\frac{m}{m-k_{\max}}\textbf{BLKdiag}\left\{\mathcal{P}_{T}\mathbf{\Sigma}_{(1)}^{-1}\mathbf{A}_{(1)}'\mathcal{P}_{\Omega_{1}^*},
     \cdots,\mathcal{P}_{T}\mathbf{\Sigma}_{(L)}^{-1}\mathbf{A}_{(L)}'\mathcal{P}_{\Omega_{L}^*}
     \right\}\right\|_{(2,2)} \nonumber\\
     &\cdot \left\|\textrm{vec}\left(\left[\mathcal{P}_{\Omega_{1}^*}\mathbf{f}_1,\cdots,\mathcal{P}_{\Omega_{L}^*}\mathbf{f}_L\right]\right)\right\|_2\nonumber\\
\leq & \left\|\frac{m}{m-k_{\max}}\textbf{BLKdiag}\left\{\mathcal{P}_{T}\mathbf{\Sigma}_{(1)}^{-1}\mathbf{A}_{(1)}'\mathcal{P}_{\Omega_{1}^*},
     \cdots,\mathcal{P}_{T}\mathbf{\Sigma}_{(L)}^{-1}\mathbf{A}_{(L)}'\mathcal{P}_{\Omega_{L}^*}
     \right\}\right\|_{(2,2)} \cdot
     \|\textrm{vec}(\mathcal{P}_{\Omega^{c}}\mathbf{F})\|_2
     \nonumber \\
\leq &
\left\|\frac{m}{m-k_{\max}}\textbf{BLKdiag}\left\{\mathcal{P}_{T}\mathbf{\Sigma}_{(1)}^{-1}\mathbf{A}_{(1)}'\mathcal{P}_{\Omega_{1}^*},
     \cdots,\mathcal{P}_{T}\mathbf{\Sigma}_{(L)}^{-1}\mathbf{A}_{(L)}'\mathcal{P}_{\Omega_{L}^*}
     \right\}\right\|_{(2,2)} \cdot \|\mathcal{P}_{\Omega^{c}}\mathbf{F}\|_1. \label{interim_appendix_3}
\end{align}

Using the definitions
$\tilde{\mathbf{A}}_{(i)}=\mathbf{\Sigma}_{(i)}^{-1}\mathbf{A}_{(i)}'\mathcal{P}_{\Omega_{i}^*}\mathbf{A}_{(i)}$
and applying the triangle inequality as well as Corollary
\ref{BEQ_corollary_plusplus}, with a high probability it holds
\begin{align*}
&\left\|\frac{m}{m-k_{\max}}\textbf{BLKdiag}\left\{\mathcal{P}_{T}\mathbf{\Sigma}_{(1)}^{-1}
\mathbf{A}_{(1)}'\mathcal{P}_{\Omega_{1}^*}\mathbf{A}_{(1)}\mathbf{\Sigma}_{(1)}^{-1}\mathcal{P}_{T},
\cdots,\mathcal{P}_{T}\mathbf{\Sigma}_{(L)}^{-1}\mathbf{A}_{(L)}'\mathcal{P}_{\Omega_{L}^*}\mathbf{A}_{(L)}\mathbf{\Sigma}_{(L)}^{-1}\mathcal{P}_{T}
     \right\}\right\|_{(2,2)}\\
=
&\left\|\frac{m}{m-k_{\max}}\textbf{BLKdiag}\left\{\mathcal{P}_{T}\tilde{\mathbf{A}}_{(1)}\mathbf{\Sigma}_{(1)}^{-1}\mathcal{P}_{T},
\cdots,\mathcal{P}_{T}\tilde{\mathbf{A}}_{(L)}\mathbf{\Sigma}_{(L)}^{-1}\mathcal{P}_{T}
     \right\}\right\|_{(2,2)}\\
\leq&\left\|
    \textbf{BLKdiag}\left\{\mathcal{P}_{T}\left(\frac{m}{m-k_{\max}}\tilde{\mathbf{A}}_{(1)}\mathbf{\Sigma}_{(1)}^{-1}-\mathbf{\Sigma}_{(1)}^{-1}\right)\mathcal{P}_{T},
    ~\cdots,~\mathcal{P}_{T}\left(\frac{m}{m-k_{\max}}\tilde{\mathbf{A}}_{(L)}\mathbf{\Sigma}_{(L)}^{-1}-\mathbf{\Sigma}_{(L)}^{-1}\right)\mathcal{P}_{T}
    \right\}
   \right\|_{(2,2)}\\
&+\left\|\textbf{BLKdiag}\left\{\mathcal{P}_{T}\mathbf{\Sigma}_{(1)}^{-1}\mathcal{P}_{T},
\cdots,\mathcal{P}_{T}\mathbf{\Sigma}_{(L)}^{-1}\mathcal{P}_{T}
     \right\}\right\|_{(2,2)}\\
\leq& \frac{\kappa_{\max}}{2}
+\left\|\textbf{BLKdiag}\left\{\mathcal{P}_{T}\mathbf{\Sigma}_{(1)}^{-1}\mathcal{P}_{T},
\cdots,\mathcal{P}_{T}\mathbf{\Sigma}_{(L)}^{-1}\mathcal{P}_{T}
     \right\}\right\|_{(2,2)}
     \leq \frac{3}{2}\kappa_{\max}.
\end{align*}
Consequently,
\[\left\|\sqrt{\frac{m}{m-k_{\max}}}\textbf{BLKdiag}\left\{\mathcal{P}_{T}\mathbf{\Sigma}_{(1)}^{-1}\mathbf{A}_{(1)}'\mathcal{P}_{\Omega_{1}^*},
     ~\cdots,~\mathcal{P}_{T}\mathbf{\Sigma}_{(L)}^{-1}\mathbf{A}_{(L)}'\mathcal{P}_{\Omega_{L}^*}\right\}
     \right\|_{(2,2)}\leq \sqrt{\frac{3}{2}\kappa_{\max}}.\]
Combining \eqref{interim_appendix_3}, this gives
\begin{equation}
\left\|\frac{m}{m-k_{\max}}\textbf{BLKdiag}\left\{\mathcal{P}_{T}\tilde{\mathbf{A}}_{(1)},
     \cdots,\mathcal{P}_{T}\tilde{\mathbf{A}}_{(L)}
     \right\}\textrm{vec}(\mathbf{H})\right\|_{2}
\leq
\sqrt{\frac{3}{2}\frac{\kappa_{\max}m}{m-k_{\max}}}\left\|\mathcal{P}_{\Omega^{c}}\mathbf{F}\right\|_1.\nonumber
\end{equation}

\noindent \textcircled{2} \textbf{Bounding the Second Term of
\eqref{interim_appendix_1}:} The following chains of equalities
and inequalities hold:
\begin{align}\label{interim_appendix_4}
      &\left\|\textbf{BLKdiag}\left\{\mathcal{P}_{T}\tilde{\mathbf{A}}_{(1)},
     \cdots,\mathcal{P}_{T}\tilde{\mathbf{A}}_{(L)}
     \right\}\textrm{vec}(\mathcal{P}_{T^c}\mathbf{H})\right\|_{2} \nonumber \\
     =&\left\|\left[\mathcal{P}_T\tilde{\mathbf{A}}_{(1)}\mathcal{P}_{T^c}\mathbf{h}_1,\cdots,\mathcal{P}_T\tilde{\mathbf{A}}_{(L)}\mathcal{P}_{T^c}\mathbf{h}_L\right]\right\|_F\nonumber\\
     =&\left\|\sum_{k\in T^c}\left[h_{1k}\mathcal{P}_T\tilde{\mathbf{A}}_{(1)}\mathbf{e}_k, \cdots,h_{Lk}\mathcal{P}_T\tilde{\mathbf{A}}_{(L)}\mathbf{e}_k\right]\right\|_F\nonumber\\
     \leq&\sum_{k\in T^c}\left\|\left[h_{1k}\mathcal{P}_T\tilde{\mathbf{A}}_{(1)}\mathbf{e}_k,\cdots,h_{Lk}\mathcal{P}_T\tilde{\mathbf{A}}_{(L)}\mathbf{e}_k\right]\right\|_F\nonumber\\
     =&\sum_{k\in T^c}\sqrt{\sum_{i=1}^L\left\|\mathcal{P}_T\tilde{\mathbf{A}}_{(i)}\mathbf{e}_k\right\|_2^2\cdot|h_{ik}|^2}\nonumber\\
     \leq&\sum_{k\in T^c}\left(\max_{i\in\{1,\cdots,L\}}\left\{\left\|\mathcal{P}_T\tilde{\mathbf{A}}_{(i)}\mathbf{e}_k\right\|_2\right\}\right)
          \cdot\left\|\mathbf{h}^k\right\|_2\nonumber\\
     \leq&\left(\max_{i\in\{1,\cdots,L\},~k\in T^c}\left\{\left\|\mathcal{P}_T\tilde{\mathbf{A}}_{(i)}\mathbf{e}_k\right\|_2\right\}\right)
          \cdot\left\|\mathcal{P}_{T^c}\mathbf{H}\right\|_{2,1},
\end{align}
where $\mathbf{h}^k$ denotes the $k$-th row of matrix $\mathbf{H}$
and $h_{ik}$ denotes the $(i,k)$-th element of $\mathbf{H}$. In
\eqref{interim_appendix_4}, the last inequality follows from the
definition of the $\ell_{2,1}$-norm. According to Lemma
\ref{lemma_BEQ_2}, with a high probability,
\eqref{interim_appendix_4} implies
\[\left\|\textbf{BLKdiag}\left\{\mathcal{P}_{T}\tilde{\mathbf{A}}_{(1)},
     \cdots,\mathcal{P}_{T}\tilde{\mathbf{A}}_{(L)}
     \right\}\textrm{vec}(\mathcal{P}_{T^c}\mathbf{H})\right\|_{2}
\leq \|\mathcal{P}_{T^c}\mathbf{H}\|_{2,1}.\]

Summarizing the results above, we have an upper bound for the
right-hand side of \eqref{interim_appendix_1}:
\begin{align}\label{interim_appendix_5}
     \|\mathcal{P}_{T}\mathbf{H}\|_F \leq& \sqrt{\frac{6\kappa_{\max}m}{m-k_{\max}}}\left\|\mathcal{P}_{\Omega^{c}}\mathbf{F}\right\|_{1}
     +\frac{2m}{m-k_{\max}}\|\mathcal{P}_{T^{c}}\mathbf{H}\|_{2,1}.
\end{align}

Finally, substituting \eqref{interim_appendix_5} into
\eqref{interim_inexact_3} gives
\begin{align}
      \left(\frac{3}{4}-\frac{1}{2\sqrt{\kappa_{\max}}}\frac{m}{m-k_{\max}}\lambda\right)\|\mathcal{P}_{T^{c}}\mathbf{H}\|_{2,1}
      +\left(\frac{3}{4}-\frac{\sqrt{6}}{4}\sqrt{\frac{m}{m-k_{\max}}}\right)\lambda\|\mathcal{P}_{\Omega^{c}}\mathbf{F}\|_1
     \leq 0. \nonumber
 \end{align}
In the above inequality,
$\frac{3}{4}-\frac{\sqrt{6}}{4}\sqrt{\frac{m}{m-k_{\max}}}$ and
$\frac{3}{4}-\frac{\lambda}{2\sqrt{\kappa_{\max}}}\frac{m}{m-k_{\max}}$
are both larger that zero provided that $\lambda<1$ and
$\frac{k_{\max}}{m}\leq\frac{\gamma}{\kappa_{\max}}<\frac{1}{3}$.
Thus, we have $\|\mathcal{P}_{T^{c}}\mathbf{H}\|_{2,1}=0$ and
$\|\mathcal{P}_{\Omega^{c}}\mathbf{F}\|_1=0$, which prove
$\mathcal{P}_{T^{c}}\mathbf{H}=\mathbf{0}$ and
$\mathcal{P}_{\Omega^{c}}\mathbf{F}=\mathbf{0}$.

\section{Proof of Theorem 4: Existence of Inexact Dual Certificate} \label{dual_construct}
\noindent \textcircled{1} \textbf{Bounding the Initial Value:}
$\left\|\lambda\mathcal{P}_{T^{c}}\left[\mathbf{A}_{(1)}'\textrm{sgn}(\bar{\mathbf{s}}_{1})
,\cdots,\mathbf{A}_{(L)}'\textrm{sgn}(\bar{\mathbf{s}}_{L})\right]\right\|_{2,\infty}\leq\frac{1}{8}$.
\begin{IEEEproof}
It is sufficient to prove
\begin{equation}\label{suff_condition_1}
\left\|\lambda\left[\mathbf{A}_{(1)}'\textrm{sgn}(\bar{\mathbf{s}}_{1})
,\cdots,\mathbf{A}_{(L)}'\textrm{sgn}(\bar{\mathbf{s}}_{L})\right]\right\|_{2,\infty}\leq\frac{1}{8}.
\end{equation}
Let $\mathbf{a}_{(i)}^r$ be the $r$-th row of
$\sqrt{m}\mathbf{A}_{(i)}'$ and $a_{(i)rj}$ be the $(r,j)$-th
element in $\sqrt{m}\mathbf{A}_{(i)}'$. Since
$\textrm{sgn}\left(\overline{\mathbf{S}}\right)$ is an i.i.d.
Rademacher random matrix (because of i.i.d. signs), for any
$r=1,\cdots,n$, we claim the following probability bound for the
row $\ell_2$-norm holds:
\begin{equation}\label{appendix_certify_1}
    Pr\left\{\sqrt{\sum_{i=1}^{L}|\mathbf{a}_{(i)}^r\textrm{sgn}(\bar{\mathbf{s}}_{i})|^{2}}
     -\sqrt{\sum_{i=1}^{L}\|\mathbf{a}_{(i)}^r\mathcal{P}_{\Omega_{i}}\|_{2}^{2}}
    \geq t\right\}\leq
    4\exp\left\{-t^{2}\left/\left(16\sum\limits_{i=1}^{L}\|\mathbf{a}_{(i)}^r\mathcal{P}_{\Omega_{i}}\|_{2}^{2}\right)\right.\right\}.
\end{equation}
The proof of \eqref{appendix_certify_1} follows from Corollary
4.10 in \cite{ledoux2001}. The details are given below.

According to Corollary 4.10 in \cite{ledoux2001}, if
$\mathbf{Z}\in\mathcal{R}^{m\times L}$ is distributed according to
some product measure on $[-1,1]^{m\times L}$ and there exists a
function $f:~\mathcal{R}^{m\times L}\rightarrow\mathcal{R}$ which
is convex and $K$-Lipschitz, then it holds
\begin{equation}
    Pr\{\left|f(\mathbf{Z})-\expect{f(\mathbf{Z})}\right|\geq t\}\leq 4\exp\left\{-\frac{t^{2}}{16K^{2}}\right\}. \nonumber
\end{equation}
Here we take
$\mathbf{Z}=\textrm{sgn}\left(\overline{\mathbf{S}}\right)$ and
$f(\cdot)=\sqrt{\sum_{i=1}^{L}|\mathbf{a}_{(i)}^r\mathcal{P}_{\Omega_i}(\cdot)|^{2}}$.
Notice that $\textrm{sgn}\left(\overline{\mathbf{S}}\right)$ is
entry-wise Bernoulli and the function $f$ we choose is convex with
the Lipschitz constant
$K\leq\sqrt{\sum_{i=1}^{L}\left\|\mathbf{a}_{(i)}^r\mathcal{P}_{\Omega_{i}}\right\|_{2}^{2}}$,
the requirements in above proposition are satisfied. In order to
bound $\expect{f(\mathbf{Z})}$ from above, we first compute
$\expect{f(\mathbf{Z})^{2}}$ and then use the property that
$\expect{f(\mathbf{Z})}\leq\sqrt{\expect{f(\mathbf{Z})^{2}}}$. We
have
\begin{align}
\expect{f(\mathbf{Z})^{2}}&=\expect{\sum_{i=1}^{L}\left|\mathbf{a}_{(i)}^r\textrm{sgn}(\overline{\mathbf{s}}_{i})\right|^{2}} \nonumber\\
                      &=\expect{\sum_{i=1}^{L}\left|\sum_{j=1}^{m}a_{(i)rj}\textrm{sgn}(\bar{s}_{ij})\right|^{2}} \nonumber\\
                      &=\sum_{i=1}^{L}\expect{\left|\sum_{j=1}^{m}a_{(i)rj}\textrm{sgn}(\bar{s}_{ij})\right|^{2}} \nonumber\\
                      &=\sum_{i=1}^{L}\sum_{j=1}^{m}\sum_{k=1}^{m} a_{(i)rj}a_{(i)rk}\textrm{sgn}(\bar{s}_{ij})\textrm{sgn}(\bar{s}_{ik}) \nonumber\\
                      &=\sum_{i=1}^{L}\left\|\mathbf{a}_{(i)}^r\mathcal{P}_{\Omega_{i}}\right\|_{2}^{2}, \nonumber
\end{align}
where the last step follows from the fact that for each
$i=1,\cdots,L$, $\textrm{sgn}(\mathbf{s}_{i})$ is a random vector
with nonzero entries i.i.d. so that all cross terms vanish. Thus,
$\expect{f(\mathbf{Z})}\leq\sqrt{\sum_{i=1}^{L}\left\|\mathbf{a}_{(i)}^r\mathcal{P}_{\Omega_{i}}\right\|_{2}^{2}}$.
Hence,
\begin{align*}
&Pr\left\{\sqrt{\sum_{i=1}^{L}\left|\mathbf{a}_{(i)}^r\textrm{sgn}(\bar{\mathbf{s}}_{i})\right|^{2}}
     -\sqrt{\sum_{i=1}^{L}\left\|\mathbf{a}_{(i)}^r\mathcal{P}_{\Omega_{i}}\right\|_{2}^{2}}
    \geq t\right\}\\
\leq&Pr\left\{\sqrt{\sum_{i=1}^{L}\left|\mathbf{a}_{(i)}^r\textrm{sgn}(\bar{\mathbf{s}}_{i})\right|^{2}}
     -\expect{f(\mathbf{Z})}
    \geq t\right\}\\
\leq&Pr\left\{\left|\sqrt{\sum_{i=1}^{L}\left|\mathbf{a}_{(i)}^r\textrm{sgn}(\bar{\mathbf{s}}_{i})\right|^{2}}
     -\expect{f(\mathbf{Z})}\right|\geq t\right\}\\
\leq&4\exp\left\{-\frac{t^{2}}{16K^2}\right\}
\leq4\exp\left\{t^{2}\left/\left(16\sum_{i=1}^{L}\left\|\mathbf{a}_{(i)}^r\mathcal{P}_{\Omega_{i}}\right\|_{2}^{2}\right)\right.\right\},
\end{align*}
which proves \eqref{appendix_certify_1}.

Next, choose $t=6\sqrt{\log
n}\sqrt{\sum_{i=1}^{L}\left\|\mathbf{a}_{(i)}^r\mathcal{P}_{\Omega_{i}}\right\|_{2}^{2}}$.
Then with probability exceeding $1-4\exp\left\{-\frac94\log
n\right\}=1-4n^{-\frac94}$, it holds
\begin{align}
    \lambda\sqrt{\sum_{i=1}^{L}\left|\mathbf{a}_{(i)}^r\textrm{sgn}(\bar{\mathbf{s}}_{i})\right|^{2}}&\leq\lambda\left(6\sqrt{\log n}+1\right)\sqrt{\sum_{i=1}^{L}\left\|\mathbf{a}_{(i)}^r\mathcal{P}_{\Omega_{i}}\right\|_{2}^{2}} \nonumber\\
&\leq7\sqrt{\sum_{i=1}^{L}\left\|\mathbf{a}_{(i)}^r\mathcal{P}_{\Omega_{i}}\right\|_{2}^{2}}
    \leq 7\sqrt{\mu_{\max} k_{\Omega}}, \nonumber
\end{align}
where the second last inequality follows from
$\lambda=\frac{1}{\sqrt{\log n}}$ and $\frac{1}{\sqrt{\log
n}}\leq1$, while the last inequality follows from the definition
of incoherence parameter in \eqref{incoherence1} and the fact that
$|\Omega|=k_{\Omega}$. Taking a union bound over all
$r=1,\cdots,n$ gives
\begin{align*}
    &Pr\left\{\left\|\lambda\mathcal{P}_{T^c}\left[\mathbf{A}_{(1)}'\textrm{sgn}(\bar{\mathbf{s}}_{1})
~\cdots~\mathbf{A}_{(L)}'\textrm{sgn}(\bar{\mathbf{s}}_{L})\right]\right\|_{2,\infty}\geq7\sqrt{\frac{\mu_{\max} k_{\Omega}}{m}}\right\}\\
=&Pr\left\{\max_{r\in\{1,2,\cdots,n\}}\left\{ \lambda\sqrt{\sum_{i=1}^{L}\left|\mathbf{a}_{(i)}^r\textrm{sgn}(\bar{\mathbf{s}}_{i})\right|^{2}}\right\}
\geq7\sqrt{\mu_{\max} k_{\Omega}}\right\}\\
\leq&\sum_{r=1}^nPr\left\{ \lambda\sqrt{\sum_{i=1}^{L}\left|\mathbf{a}_{(i)}^r\textrm{sgn}(\bar{\mathbf{s}}_{i})\right|^{2}}
\geq7\sqrt{\mu_{\max} k_{\Omega}}\right\}\\
\leq&4n^{-\frac94}\cdot n=4n^{-\frac54}\leq4n^{-1},
\end{align*}
where the first equality follows from the definition of
$\mathbf{a}_{(i)}^r$. Substituting the bounds
$k_{\Omega}\leq\beta\frac{m}{\mu_{\max}}$ and
$\beta\leq\frac{1}{3136}$ into the
above inequality finally gives
\[Pr\left\{\left\|\lambda\left[\mathbf{A}_{(1)}'\textrm{sgn}(\bar{\mathbf{s}}_{1})
,\cdots,\mathbf{A}_{(L)}'\textrm{sgn}(\bar{\mathbf{s}}_{L})\right]\right\|_{2,\infty}\geq\frac18\right\}\leq4n^{-1},\]
which finishes the proof.
\end{IEEEproof}

\noindent \textcircled{2} \textbf{Bounding the term:}
$\left\|\mathcal{P}_{T}\mathbf{U}
+\lambda\mathcal{P}_{T}\left[\mathbf{A}_{(1)}'\textrm{sgn}(\bar{\mathbf{s}}_{1})
,\cdots,\mathbf{A}_{(L)}'\textrm{sgn}(\bar{\mathbf{s}}_{L})\right]
-\overline{\mathbf{V}}\right\|_{F}\leq\frac{\lambda}{4\sqrt{\kappa_{\max}}}$.
\begin{IEEEproof}
Recalling the definition of $\mathbf{U}$ in
\eqref{definition_dual_U}, we have
\begin{align}
    \mathcal{P}_{T}\mathbf{U}
    =&\mathcal{P}_{T}\left[\sum_{j=1}^{l}\frac{m}{m_{j}}\tilde{\mathbf{A}}_{(1,j)}\mathcal{P}_{T}\mathbf{q}_{(j-1)1}
    ,\cdots
    ,\sum_{j=1}^{l}\frac{m}{m_{j}}\tilde{\mathbf{A}}_{(L,j)}\mathcal{P}_{T}\mathbf{q}_{(j-1)L}\right]. \nonumber
\end{align}
According to the definition of $\mathbf{Q}_{(0)}$ in
\eqref{definition_Q_0},
$\mathcal{P}_T\mathbf{Q}_{(0)}=\mathbf{Q}_{(0)}$. Since each
subsequent mapping from $\mathbf{Q}_{(j-1)}$ to $\mathbf{Q}_{(j)}$
defined in \eqref{definition_recursive} is a mapping from $T$ to
$T$, it follows that
$\mathcal{P}_T\mathbf{Q}_{(j)}=\mathbf{Q}_{(j)}$ for any
$j=1,\cdots,l$. Therefore, it holds
\begin{align}
    \mathcal{P}_T\mathbf{U}=&\sum_{j=1}^l\left(\mathbf{Q}_{(j-1)}-\mathcal{P}_T\mathbf{Q}_{(j-1)}\right)
               +\mathcal{P}_{T}\left[\sum_{j=1}^{l}\frac{m}{m_{j}}\tilde{\mathbf{A}}_{(1,j)}\mathcal{P}_{T}\mathbf{q}_{(j-1)1}
               ,\cdots
               ,\sum_{j=1}^{l}\frac{m}{m_{j}}\tilde{\mathbf{A}}_{(L,j)}\mathcal{P}_{T}\mathbf{q}_{(j-1)L}\right]\nonumber\\
    =&\sum_{j=1}^{l}\left(\mathbf{Q}_{(j-1)}-\left[\mathcal{P}_{T}\left(\mathbf{I}-\frac{m}{m_{j}}\tilde{\mathbf{A}}_{(1,j)}\right)\mathcal{P}_{T}\mathbf{q}_{(j-1)1}
    ,\cdots,\mathcal{P}_{T}\left(\mathbf{I}-\frac{m}{m_{j}}\tilde{\mathbf{A}}_{(L,j)}\right)\mathcal{P}_{T}\mathbf{q}_{(j-1)L} \right]\right) \nonumber\\
    =&\sum_{j=1}^{l}\left(\mathbf{Q}_{(j-1)}-\mathbf{Q}_{(j)}\right)=\mathbf{Q}_{(0)}-\mathbf{Q}_{(l)},\nonumber
\end{align}
where the second last equality follows from the definition of
$\mathbf{Q}_{(j)}$. Thus, substituting the definition of
$\mathbf{Q}_{(0)}$ in \eqref{definition_Q_0} yields
\begin{align*}
\mathbf{Q}_{(l)}=&\mathbf{Q}_{(0)}-\mathcal{P}_T\mathbf{U}
=\overline{\mathbf{V}}-\lambda\mathcal{P}_{T}\left[\mathbf{A}_{(1)}'\textrm{sgn}(\bar{\mathbf{s}}_{1})
,\cdots,\mathbf{A}_{(L)}'\textrm{sgn}(\bar{\mathbf{s}}_{L})\right]
-\mathcal{P}_T\mathbf{U},
\end{align*}
which further implies
\begin{equation}
    \left\|\mathcal{P}_{T}\mathbf{U} +\lambda\mathcal{P}_{T}\left[\mathbf{A}_{(1)}'\textrm{sgn}(\bar{\mathbf{s}}_{1})
,\cdots,\mathbf{A}_{(L)}'\bar{sgn}(\bar{\mathbf{s}}_{L})\right]
-\overline{\mathbf{V}}\right\|_{F} =\|\mathbf{Q}_{(l)}\|_{F}.
\nonumber
\end{equation}
Thus, we are able to bound the target function on the left-hand
side by bounding $\|\mathbf{Q}_{(l)}\|_{F}$ instead. It is enough
to obtain an upper bound for $\|\mathbf{Q}_{(0)}\|_{F}$ and apply
contractions \eqref{Q_contraction_1}-\eqref{Q_contraction_2}. From
\eqref{suff_condition_1}, it follows
\begin{equation}
    \left\|\lambda\mathcal{P}_{T}\left[\mathbf{A}_{(1)}'\textrm{sgn}(\bar{\mathbf{s}}_{1})
,\cdots,\mathbf{A}_{(L)}'\textrm{sgn}(\bar{\mathbf{s}}_{L})\right]\right\|_{F}\leq\frac{\sqrt{k_{T}}}{8}.
\nonumber
\end{equation}
Since $\|\overline{\mathbf{V}}\|_F=\sqrt{k_T}$, by triangle
inequality, we have
\begin{equation}\label{Q_0_bound}
    \left\|\mathbf{Q}_{(0)}\right\|=\left\|\lambda\mathcal{P}_{T}\left[\mathbf{A}_{(1)}'\textrm{sgn}(\bar{\mathbf{s}}_{1})
,\cdots,\mathbf{A}_{(L)}'\textrm{sgn}(\bar{\mathbf{s}}_{L})\right]
-\overline{\mathbf{V}}\right\|_{F}\leq\frac{9\sqrt{k_{T}}}{8}.
\end{equation}
Thus, by contractions of $\{\mathbf{Q}_{(j)}\}_{j=1}^l$ in
\eqref{Q_contraction_1}-\eqref{Q_contraction_2}, we have
$$\|\mathbf{Q}_{(l)}\|_{F}\leq\frac{1}{\log
n}\frac{1}{2^{l}}\|\mathbf{Q}_{(0)}\|_{F} \leq\frac{1}{\log
n}\frac{1}{2^{l}}\frac{9\sqrt{k_{T}}}{8}\leq\frac{1}{\log
n}\frac{1}{n}\frac{9\sqrt{k_{T}}}{8}
\leq\frac{\lambda}{4\sqrt{\kappa_{max}}},$$ provided that
$\alpha\leq\frac{4}{81}$, where the last inequality follows from
the fact $k_{T}\leq \alpha\frac{m}{\mu_{\max}\kappa_{\max}\log
n}$, and $\sqrt{m}\leq \sqrt{\mu_{\max}}n\log n$. Furthermore,
from the proof, as long as \eqref{batch_inequalities} and
\eqref{suff_condition_1} hold, this bound is guaranteed.
\end{IEEEproof}

\noindent \textcircled{3} \textbf{Bounding the term:}
$\|\mathcal{P}_{T^{c}}\mathbf{U}\|_{2,\infty}\leq\frac{1}{8}$.
\begin{IEEEproof}
We claim that the following inequality is true with high
probability:
\begin{equation}\label{appendix_certificate_3_1}
\|\mathcal{P}_{T^{c}}\mathbf{U}\|_{2,\infty}\leq\sum_{j=1}^{l}\frac{1}{10\sqrt{k_{T}}}\|\mathbf{Q}_{(j-1)}\|_{F}.
\end{equation}

According to the definition of $\mathbf{U}$ in
\eqref{definition_dual_U}, it holds
\begin{align*}
    \mathcal{P}_{T^{c}}\mathbf{U}
    &=\left[\sum_{j=1}^{l}\frac{m}{m_{j}}\mathcal{P}_{T^{c}}\tilde{\mathbf{A}}_{(1,j)}\mathcal{P}_{T}\mathbf{q}_{(j-1)1}
    ,\cdots,\sum_{j=1}^{l}\frac{m}{m_{j}}\mathcal{P}_{T^{c}}\tilde{\mathbf{A}}_{(L,j)}\mathcal{P}_{T}\mathbf{q}_{(j-1)L}\right]\\
    &=\sum_{j=1}^{l}\left[\frac{m}{m_{j}}\mathcal{P}_{T^{c}}\tilde{\mathbf{A}}_{(1,j)}\mathcal{P}_{T}\mathbf{q}_{(j-1)1}
    ,\cdots,\frac{m}{m_{j}}\mathcal{P}_{T^{c}}\tilde{\mathbf{A}}_{(L,j)}\mathcal{P}_{T}\mathbf{q}_{(j-1)L}\right].
\end{align*}
Thus, it is enough to show that for any $k\in T^c$, it holds
\begin{equation}\label{appendix_certificate_3_2}
    \left\|\sum_{j=1}^{l}\left[\frac{m}{m_{j}}\mathbf{e}_{k}'\tilde{\mathbf{A}}_{(1,j)}\mathcal{P}_{T}\mathbf{q}_{(j-1)1}
    ,\cdots,\frac{m}{m_{j}}\mathbf{e}_{k}'\tilde{\mathbf{A}}_{(L,j)}\mathcal{P}_{T}\mathbf{q}_{(j-1)L}\right]\right\|_{2}
    \leq\sum_{j=1}^{l}\frac{1}{10\sqrt{k_{T}}}\|\mathbf{Q}_{(j-1)}\|_{F},
\end{equation}
with a high probability, where $\{\mathbf{e}_k\}_{k=1}^n$ is a
standard basis in $\mathcal{R}^n$. By the triangle inequality, a
sufficient condition for \eqref{appendix_certificate_3_2} to
satisfy is
\begin{equation}
\sum_{j=1}^{l}\left\|\left[\frac{m}{m_{j}}\mathbf{e}_{k}'\tilde{\mathbf{A}}_{(1,j)}\mathcal{P}_{T}\mathbf{q}_{(j-1)1}
~~\cdots~~\frac{m}{m_{j}}\mathbf{e}_{k}'\tilde{\mathbf{A}}_{(L,j)}\mathcal{P}_{T}\mathbf{q}_{(j-1)L}\right]\right\|_{2}
\leq\sum_{j=1}^{l}\frac{1}{10\sqrt{k_{T}}}\|\mathbf{Q}_{(j-1)}\|_{F}.
\nonumber
\end{equation}
Therefore, it resorts to proving a one-step-further sufficient
condition that with a high probability, for any $j=1,\cdots,l$ and
$k\in T^c$, it holds
\begin{equation}\label{appendix_certificate_3_3}
\left\|\left[\frac{m}{m_{j}}\mathbf{e}_{k}'\tilde{\mathbf{A}}_{(1£¬j)}\mathcal{P}_{T}\mathbf{q}_{(j-1)1}
~~\cdots~~\frac{m}{m_{j}}\mathbf{e}_{k}'\tilde{\mathbf{A}}_{(L,j)}\mathcal{P}_{T}\mathbf{q}_{(j-1)L}\right]\right\|_{2}
\leq\frac{1}{10\sqrt{k_{T}}}\|\mathbf{Q}_{(j-1)}\|_{F}.
\end{equation}

We apply the vector Berstein inequality to prove
\eqref{appendix_certificate_3_3}. First, for any $i=1,\cdots,L$,
and any $r\in K_{ij}$, let
$$g_{(i,r)}=\frac{1}{m_{j}}\mathbf{e}_{k}'\mathbf{\Sigma}_{(i)}^{-1}\mathbf{a}_{(i)r}\mathbf{a}_{(i)r}'
\mathcal{P}_{T}\mathbf{q}_{(j-1)i}.$$ Observe the fact that
$\{\mathbf{a}_{(i)j}\}_{j\in K_{ij}}$ is the set of column vectors
in $\sqrt{m}\mathbf{A}_{(i)}'P_{K_{ij}}$, which are nonzero. Also,
recall the definition
$\tilde{\mathbf{A}}_{(i,j)}=\mathbf{\Sigma}_{(i)}^{-1}\mathbf{A}_{(i)}'\mathcal{P}_{K_{ij}}\mathbf{A}_{(i)}$.
For any $i=1,\cdots,L$, it follows
\begin{equation}\label{appendix_certificate_3_extra}
\sum_{r\in K_{ij}}g_{(i,r)}=\frac{m}{m_{j}}\mathbf{e}_{k}'\tilde{\mathbf{A}}_{(i,j)}\mathcal{P}_{T}\mathbf{q}_{(j-1)i}.
\end{equation}
For notation convenience, without loss of generality, suppose
$K_{ij}=\{1,\cdots,m_j\}$, $\forall i=1,\cdots,L$. For any
$r=1,\cdots,m_j$, we align the scalars $g_{(i,r)}$, $i=1,\cdots,L$
into a single vector as $\left[g_{(1,r)},\cdots,g_{(L,r)}\right]$.
According to \eqref{appendix_certificate_3_extra}, this vector
satisfies
\[\sum_{r=1}^{m_j}\left[g_{(1,r)},\cdots,g_{(L,r)}\right]
=\left[\frac{m}{m_{j}}\mathbf{e}_{k}'\tilde{\mathbf{A}}_{(1,j)}\mathcal{P}_{T}\mathbf{q}_{(j-1)1}
,\cdots,\frac{m}{m_{j}}\mathbf{e}_{k}'\tilde{\mathbf{A}}_{(L,j)}\mathcal{P}_{T}\mathbf{q}_{(j-1)L}\right].\]

Notice that $\mathbf{Q}_{(j-1)}$ is also a random variable. In the
following proof, we apply the vector Berstein inequality
conditioned on $\mathbf{Q}_{(j-1)}$. It is obvious that given
$\mathbf{Q}_{(j-1)}$, $\left[g_{(1,r)},\cdots,g_{(L,r)}\right]$
are i.i.d. for different $r$ and
$$\expect{\left[g_{(1,r)},\cdots,g_{(L,r)}\right]\left|\mathbf{Q}_{(j-1)}\right.}
=\frac{m}{m_{j}}\left[\mathbf{e}_{k}'\mathcal{P}_{T}\mathbf{q}_{(j-1)1},\cdots,\mathbf{e}_{k}'\mathcal{P}_{T}\mathbf{q}_{(j-1)L}\right]=0,$$
since $k\in T^c$. Next, we compute
\begin{align}
    \expect{\left.\left(g_{(i,r)}\right)^2\right|\mathbf{Q}_{(j-1)}}
    =&\frac{1}{m_{j}^2}\expect{\left.\left(\mathbf{e}_{k}'\mathbf{\Sigma}_{(i)}^{-1}\mathbf{a}_{(i)r}\mathbf{a}_{(i)r}'
    \mathcal{P}_{T}\mathbf{q}_{(j-1)i}\right)^2\right|\mathbf{Q}_{(j-1)}} \nonumber\\
    =&\frac{1}{m_{j}^2}\expect{\left.(\mathbf{e}_{k}'\mathbf{\Sigma}_{(i)}^{-1}\mathbf{a}_{(i)r})^2
     (\mathbf{a}_{(i)r}'\mathcal{P}_{T}\mathbf{q}_{(j-1)i})^2\right|\mathbf{Q}_{(j-1)}}\nonumber\\
    \leq&\frac{\mu_{i} }{m_{j}^2}\expect{\left.\mathbf{q}_{(j-1)i}'\mathcal{P}_{T}\mathbf{a}_{(i)r}
         \mathbf{a}_{(i)r}'\mathcal{P}_{T}\mathbf{q}_{(j-1)i}\right|\mathbf{Q}_{(j-1)}}
    \nonumber\\
    \leq& \frac{\mu_{i}\kappa_{i}}{m_j^2}\|\mathbf{q}_{(j-1)i}\|_{2}^{2}. \nonumber
\end{align}
Therein, the first inequality follows from the definition of the
incoherence parameter \eqref{incoherence2}. The second inequality
follows from the fact that for each $i=1,\cdots,L$, the sampled
batches of vectors $\mathcal{P}_{K_{ij}}\mathbf{A}_{(i)}$ are not
overlapped for different batches $j=1,\cdots,l$ such that
$\mathbf{q}_{(j-1)i}$ and $\mathbf{a}_{(i,r)}$ are independent.
Thus, we have
\begin{equation}
    \sum_{r=1}^{m_j}\expect{\left\|\left[g_{(1,r)},\cdots,g_{(L,r)}\right]\right\|_{2}^2\left|\mathbf{Q}_{(j-1)}\right.}
    \leq\sum_{r=1}^{m_j}\sum_{i=1}^L\frac{\mu_{i}\kappa_{i}}{m_j^2}\|\mathbf{q}_{(j-1)i}\|_{2}^{2}
    \leq \frac{\mu_{\max}\kappa_{\max}}{m_j}\|\mathbf{Q}_{(j-1)}\|_{F}^{2}\triangleq\sigma^2. \nonumber
\end{equation}
Moreover,
\begin{equation}
    |g_{(i,r)}|=\frac{m}{m_{j}}\left|\mathbf{e}_{k}'\mathbf{\Sigma}_{(i)}^{-1}\mathbf{a}_{(i)r}\mathbf{a}_{(i)r}'\mathcal{P}_{T}\mathbf{q}_{(j-1)i}\right|
    \leq\frac{m}{m_{j}}\sqrt{\mu_i}\left|\mathbf{a}_{(i)r}'\mathcal{P}_{T}\mathbf{q}_{(j-1)i}\right|
    \leq\frac{\mu_{i}}{m_j}\|\mathbf{q}_{(j-1)i}\|_{2}, \nonumber
\end{equation}
where the first inequality follows from the incoherence assumption
\eqref{incoherence2} and the second inequality follows from the
incoherence condition \eqref{incoherence1}. Thus, it holds
\begin{equation}
    \left\|\left[g_{(1,r)},\cdots,g_{(L,r)}\right]\right\|_2\leq\frac{\mu_{\max}}{m_j}\|\mathbf{Q}_{(j-1)}\|_{F}\triangleq B.\nonumber
\end{equation}
Substituting the above bounds into the vector Berstein inequality
conditioned on $\mathbf{Q}_{(j-1)}$ gives
\begin{equation}\label{appendix_certificate_3_4}
Pr\left(\left.\left\|\sum_{r=1}^{m_j}\left[g_{(1,r)},\cdots,g_{(L,r)}\right]\right\|_{2}\geq
t\right|\mathbf{Q}_{(j-1)}\right)\leq
\textrm{exp}\left(-\frac{t^2}{\frac{8\mu_{\max}\kappa_{\max}}{m_j}\|\mathbf{Q}_{(j-1)}\|_{F}^{2}}+\frac{1}{4}\right).
\end{equation}
We choose $t=\sqrt{\frac{24}{m_{j}}\mu_{\max} \kappa_{\max}\log
n}\|\mathbf{Q}_{(j-1)}\|_{F}$. First, we need to verify that such
a choice satisfies $t\leq\frac{\sigma^2}{B}$. Recall that for any
$j=1,\cdots,l$, $m_j\geq\frac{m}{4\log n}$. Since
$k_T\leq\alpha\frac{m}{\mu_{\max}\kappa_{\max}\log^2n}$ and
$\alpha\leq\frac{1}{9600}$, it holds
$$t\leq\frac{1}{10\sqrt{k_T}}\|\mathbf{Q}_{(j-1)}\|_{F}\leq \kappa_{\max}\|\mathbf{Q}_{(j-1)}\|_{F}=\frac{\sigma^2}{B}.$$
Thus, the choice of $t$ is indeed valid. Substituting this $t$
into \eqref{appendix_certificate_3_4} gives
\begin{equation}
Pr\left(\left.\left\|\sum_{r=1}^{m_j}\left[g_{(1,r)},\cdots,g_{(L,r)}\right]\right\|_{2}
\geq \sqrt{\frac{24}{m_{j}}\mu_{\max} \kappa_{\max}\log
n}\|\mathbf{Q}_{(j-1)}\|_{F}\right|\mathbf{Q}_{(j-1)}\right) \leq
e^{\frac{1}{4}}n^{-3}, \nonumber
\end{equation}
which implies
\begin{equation}
Pr\left(\left.\left\|\sum_{r=1}^{m_j}\left[g_{(1,r)},\cdots,g_{(L,r)}\right]\right\|_{2}
\geq
\frac{1}{10\sqrt{k_T}}\|\mathbf{Q}_{(j-1)}\|_{F}\right|\mathbf{Q}_{(j-1)}\right)
\leq e^{\frac{1}{4}}n^{-3}. \nonumber
\end{equation}
Since the right-hand side does not depend on $\mathbf{Q}_{(j-1)}$,
taking expectation from both sides regarding $\mathbf{Q}_{(j-1)}$
gives
\begin{equation}
Pr\left(\left\|\sum_{r=1}^{m_j}\left[g_{(1,r)},\cdots,g_{(L,r)}\right]\right\|_{2}
\geq \frac{1}{10\sqrt{k_T}}\|\mathbf{Q}_{(j-1)}\|_{F}\right) \leq
e^{\frac{1}{4}}n^{-3}. \nonumber
\end{equation}
Take a union bound over all $j=1,\cdots,l$ and $k\in T^c$ gives
\begin{align*}
&Pr\left(\max_{j\in\{1,\cdots,l\},k\in
T^c}\left\{\left\|\sum_{r=1}^{m_j}\left[g_{(1,r)},\cdots,g_{(L,r)}\right]\right\|_{2}\right\}
\geq \frac{1}{10\sqrt{k_T}}\|\mathbf{Q}_{(j-1)}\|_{F}\right)\\
\leq&\sum_{j=1}^l\sum_{k\in
T^c}Pr\left(\left\|\sum_{r=1}^{m_j}\left[g_{(1,r)},\cdots,g_{(L,r)}\right]\right\|_{2}
\geq \frac{1}{10\sqrt{k_T}}\|\mathbf{Q}_{(j-1)}\|_{F}\right)\\
\leq&\sum_{j=1}^l\sum_{k\in T^c}e^{-\frac{1}{4}}n^{-3} \leq (\log
n+1)\cdot n \cdot e^{\frac{1}{4}}n^{-3}\leq e^{\frac{1}{4}}n^{-1}.
\end{align*}
This proves \eqref{appendix_certificate_3_3} and further implies
that \eqref{appendix_certificate_3_1} holds. Finally, applying the
contractions \eqref{Q_contraction_1}-\eqref{Q_contraction_2} gives
\begin{align*}
    \|\mathcal{P}_{T^{c}}\mathbf{U}\|_{2,\infty}
    \leq \sum_{j=1}^{l}\frac{1}{10\sqrt{k_{T}}}\|\mathbf{Q}_{(j-1)}\|_{F}
    \leq \sum_{j=1}^{l}\frac{1}{10\sqrt{k_{T}}}\frac{1}{2^{j}}\|\mathbf{Q}_{(0)}\|_{F}
    \leq \frac{1}{10\sqrt{k_{T}}}\|\mathbf{Q}_{(0)}\|_{F}.
\end{align*}
Substituting the bound on $\|\mathbf{Q}_{(0)}\|_{F}$ in
\eqref{Q_0_bound} gives the desired result. Notice that the
inequality
$\|\mathcal{P}_{T^{c}}\mathbf{U}\|_{2,\infty}\leq\frac18$ requires
\eqref{batch_inequalities}, \eqref{suff_condition_1} and
\eqref{appendix_certificate_3_1} to hold simultaneously.
\end{IEEEproof}

\noindent \textcircled{4} \textbf{Bounding the term:}
$\|\mathcal{P}_{\Omega^{c}}\mathbf{W}\|_{\infty}\leq\frac{\lambda}{4}$.
\begin{IEEEproof}
According to the definition of $\mathbf{W}$ in \eqref{definition_dual_W}, we aim to prove
\begin{equation}
   \left\| \left[
    \begin{array}{ccc}
        \sum_{j=1}^{l}\frac{m}{m_{j}}\mathcal{P}_{K_{1j}}\mathbf{A}_{(1)}\mathcal{P}_{T}\mathbf{q}_{(j-1)1}, & \cdots, & \sum_{j=1}^{l}\frac{m}{m_{j}}\mathcal{P}_{K_{Lj}}\mathbf{A}_{(L)}\mathcal{P}_{T}\mathbf{q}_{(j-1)L}
    \end{array}
    \right] \right\|_{\infty}\leq\frac{\lambda}{4}. \nonumber
\end{equation}
Notice that the batch sets $K_{ij}$, $j=1,\cdots,l$ are not
overlapped. Therefore, it is enough to show with a high probability,
for any $i=1,\cdots,L$, any $j=1,\cdots,l$, and any vector
$\mathbf{a}_{(i)r}$ with $r\in K_{ij}$, it holds
\[\left| \frac{\sqrt{m}}{m_j}\mathbf{a}_{(i)r}'\mathcal{P}_{T}\mathbf{q}_{(j-1)i} \right|\leq\frac{\lambda}{4}.\]
Equivalently, according to the definition of $\mathbf{Q}_{(j)}$ in
\eqref{definition_recursive}, it is enough to prove for $j\geq2$
it holds
\begin{equation}
    \left| \frac{\sqrt{m}}{m_j}\mathbf{a}_{(i)r}'\left(\prod_{k=1}^{j-1}\mathcal{P}_{T}\left(\mathbf{I}-\tilde{\mathbf{A}}_{(1,k)}\right)\mathcal{P}_{T}\right)
    \mathbf{q}_{(0)i} \right|\leq\frac{\lambda}{4},\nonumber
\end{equation}
and for $j=1$ it holds
\begin{equation}
    \left| \frac{\sqrt{m}}{m_j}\mathbf{a}_{(i)r}'\mathcal{P}_{T}\mathbf{q}_{(0)i} \right|\leq\frac{\lambda}{4}.\nonumber
\end{equation}
In order to further simplify the notation, for any vector $\mathbf{a}_{(i)r}$ such that $r\in K_{ij}$, let
\begin{equation}
    \mathbf{g}_{(i,r)}'\triangleq\left\{
                     \begin{array}{ll}
                       \mathbf{a}_{(i)r}'\left(\prod_{k=1}^{j-1}\mathcal{P}_{T}\left(\mathbf{I}-\tilde{\mathbf{A}}_{(1,k)}\right)\mathcal{P}_{T}\right)
                                          , & \hbox{if $j\geq2$;} \\
                       \mathbf{a}_{(i)r}', & \hbox{if $j=1$.}
                     \end{array}
                   \right.\nonumber
\end{equation}

Our goal is to prove that for any $i=1,\cdots,L$, any
$j=1,\cdots,l$, and any vector $\mathbf{a}_{(i,r)}$ in the $j$-th
batch vectors $\mathcal{P}_{K_{ij}}A_{(i)}$, with a high probability
it holds
\begin{equation}\label{appendix_certificate_4_1}
\left| \frac{\sqrt{m}}{m_j}\mathbf{g}_{(i,r)}'\mathbf{q}_{(0)i} \right|\leq\frac{\lambda}{4}.
\end{equation}
Since both $\mathbf{g}_{(i,r)}$ and $\mathbf{q}_{(0)i}$ are random
variables, it is easier to first bound the left-hand side of
\eqref{appendix_certificate_4_1} conditioned on
$\mathbf{g}_{(i,r)}$. Recall the definition of $\mathbf{Q}_{(0)}$
in \eqref{definition_Q_0}, for any $i=1,\cdots,L$, it holds
\[\mathbf{q}_{(0)i}=\bar{\mathbf{v}}_i-\mathcal{P}_{T}\mathbf{A}_{(i)}'\textrm{sgn}(\mathbf{s}_{i}).\]
By the triangle inequality, we bound
$\left|\mathbf{g}_{(i,r)}'\bar{\mathbf{v}}_i\right|$ and
$\left|\mathbf{g}_{(i,r)}'\mathcal{P}_{T}\mathbf{A}_{(i)}'\textrm{sgn}(\mathbf{s}_{i})\right|$,
respectively.

Let us first bound
$\left|\mathbf{g}_{(i,r)}'\bar{\mathbf{v}}_i\right|$ conditioned
on $\mathbf{g}_{(i,r)}$. From our assumption, the vector
$\bar{\mathbf{v}}_i$ is fixed except for the i.i.d. signs. Denote
$\left|\bar{\mathbf{v}}_i\right|$ as the \emph{entry-wise absolute
value vector} of $\bar{\mathbf{v}}_i$, which is not random. Then,
\[\mathbf{g}_{(i,r)}'\bar{\mathbf{v}}_i=\left(\mathbf{g}_{(i,r)}\odot\left|\bar{\mathbf{v}}_i\right|\right)'
\cdot\textrm{sgn}(\bar{\mathbf{v}}_i),\] where $\odot$ denotes the
entry-wise Hadamand product. Notice that
$\textrm{sgn}(\bar{\mathbf{v}}_i)$ and $\mathbf{g}_{(i,r)}$ are
mutually independent. Applying the Hoeffding inequality conditioned on
$\mathbf{g}_{(i,r)}$ gives
\begin{equation}
    Pr\left\{\left.\left|\left(\mathbf{g}_{(i,r)}\odot\left|\bar{\mathbf{v}}_i\right|\right)'
    \cdot\textrm{sgn}(\bar{\mathbf{v}}_i)\right|\geq t
    ~~\right|~~\mathbf{g}_{(i,r)}\right\}
    \leq2\exp\left\{-\frac{t^{2}}{2\left\|\mathbf{g}_{(i,r)}\odot\left|\bar{\mathbf{v}}_i\right|\right\|_{2}^{2}}\right\}.\nonumber
\end{equation}
Since each entry of $\bar{\mathbf{v}}_i$ is within $[-1,1]$, by
taking $t=2\sqrt{\log n}\|\mathbf{g}_{(i,r)}\|_2$, it follows
\begin{equation}\label{appendix_certificate_4_2}
Pr\left\{\left.\left|\left(\mathbf{g}_{(i,r)}\odot\left|\bar{\mathbf{v}}_i\right|\right)'
    \cdot\textrm{sgn}(\bar{\mathbf{v}}_i)\right|\geq 2\sqrt{\log n}\|\mathbf{g}_{(i,r)}\|_2
    ~~\right|~~\mathbf{g}_{(i,r)}\right\}\leq2n^{-2}.
\end{equation}

Second, we bound
$\left|\mathbf{g}_{(i,r)}'\mathcal{P}_{T}\mathbf{A}_{(i)}'\textrm{sgn}(\bar{\mathbf{s}}_{i})\right|$
conditioned on $\mathbf{g}_{(i,r)}$. The key is to prove the
argument that $\mathcal{P}_{T}\mathbf{g}_{(i,r)}$ is independent
of
$\mathcal{P}_{T}\mathbf{A}_{(i)}'\textrm{sgn}(\bar{\mathbf{s}}_{i})$.
Notice that by definition, $\mathbf{g}_{(i,r)}$ is generated by
the column vectors in $\mathbf{A}_{(i)}'$ with column indices from
the batch sets $K_{ij}$, $j=1,\cdots,l$. Recall the definition of
these batch sets under \eqref{definition_Q_0},
$\cup_{j=1}^lK_{ij}\subseteq\Omega_i^*\subseteq\Omega_i^c$. On the
other hand, $\mathbf{A}_{(i)}'\textrm{sgn}(\mathbf{s}_{i})$ picks
out those column vectors in $\mathbf{A}_{(i)}$ with the column
indices from $\Omega_i$. Since different columns of
$\mathbf{A}_{(i)}'$ are i.i.d. samples from the distribution
$\mathcal{F}_i$, the argument holds true.

Moreover, since the noise support $\Omega_i$ are assumed to be fixed and the signs of noise matrix are i.i.d., $\mathbf{A}_{(i)}$ and $\textrm{sgn}(\bar{\mathbf{s}}_{i})$ are also independent. We write
$$\mathbf{g}_{(i,r)}'\mathcal{P}_{T}\mathbf{A}_{(i)}'\textrm{sgn}(\bar{\mathbf{s}}_{i})
=\frac{1}{\sqrt{m}}\sum_{x\in\Omega_i}\mathbf{g}_{(i,r)}'\mathbf{a}_{(i)x}\cdot\textrm{sgn}(\bar{s}_{ix}).$$
Then, for any $x\in\Omega_i$, we have
\begin{align*}
    &\expect{\left.\mathbf{g}_{(i,r)}'\mathcal{P}_{T}\mathbf{a}_{(i)x}\textrm{sgn}(\bar{s}_{ix})\right|\mathbf{g}_{(i,r)}}
     =\mathbf{g}_{(i,r)}'\expect{\mathcal{P}_{T}\mathbf{a}_{(i)x}'}\expect{\textrm{sgn}(\bar{s}_{ix})}=0, \\
    &\left|\mathbf{g}_{(i,r)}'\mathcal{P}_{T}\mathbf{a}_{(i)x}\textrm{sgn}(\bar{s}_{ix})\right|
     \leq\sqrt{\mu_i k_{T}}\|\mathbf{g}_{(i,r)}\|_{2}, \\
    &\expect{\left.\left|\mathbf{g}_{(i,r)}'\mathcal{P}_{T}\mathbf{a}_{(i)x}\textrm{sgn}(\bar{s}_{ix})\right|^2\right|\mathbf{g}_{(i,r)}}
      =\mathbf{g}_{(i,r)}'\expect{\mathcal{P}_{T}\mathbf{a}_{(i)x}\mathbf{a}_{(i)x}'\mathcal{P}_{T}}\mathbf{g}_{(i,r)}
       \leq\kappa_{i}\|\mathbf{g}_{(i,r)}\|_{2}^{2}.
\end{align*}
Thus, using the one dimensional Berstein inequality (which can
also be regarded as a special case of the matrix Berstein
inequality), we have
\begin{align*}
     &Pr\left\{\left.\left|\mathbf{g}_{(i,r)}'\mathcal{P}_{T}\mathbf{A}_{(i)}'\textrm{sgn}(\bar{\mathbf{s}}_{i})\right|
      >\frac{t}{\sqrt{m}}\right|\mathbf{g}_{(i,r)}\right\}\\
    =&Pr\left\{\left.\left|\sum_{x\in\Omega_i}\mathbf{g}_{(i,r)}'\mathbf{a}_{(i)x}\cdot\textrm{sgn}(\bar{s}_{ix})\right|
      >t\right|\mathbf{g}_{(i,r)}\right\}\\
    \leq&2\exp\left(-\frac{\frac{1}{2}t^{2}}{k_{\Omega_i}\kappa_{i}\|\mathbf{g}_{(i,r)}\|_{2}^{2}
         +\sqrt{k_{T}\mu_{i}}\frac{\|\mathbf{g}_{(i,r)}\|_{2}t}{3}}\right).
\end{align*}
Since $k_{\max}\leq \gamma \frac{m}{\kappa_{\max}}$ with
$\gamma\leq\frac{1}{4}$ and
$k_{T}\leq\alpha\frac{m}{\mu_{\max}\kappa_{\max}\log^2n}$ with
$\alpha\leq\frac{1}{9600}$, choosing
$t=2\sqrt{m\log{n}}\|\mathbf{g}_{(i,r)}\|_{2}$ gives
\begin{align}\label{appendix_certificate_4_3}
&Pr\left\{\left.\left|\mathbf{g}_{(i,r)}'\mathcal{P}_{T}\mathbf{A}_{(i)}'\textrm{sgn}(\bar{\mathbf{s}}_{i})\right|
      >2\sqrt{\log{n}}\|\mathbf{g}_{(i,r)}\|_{2}\right|\mathbf{g}_{(i,r)}\right\}\nonumber\\
\leq&2\exp\left\{-\frac{2m\log n}{\frac{m}{4}+\frac{1}{60\sqrt{6\log n}}}\right\}
\leq2n^{-2}.
\end{align}
Combining \eqref{appendix_certificate_4_2} and \eqref{appendix_certificate_4_3} gives
\begin{align*}
&Pr\left\{\left.\left| \frac{\sqrt{m}}{m_j}\mathbf{g}_{(i,r)}'\mathbf{q}_{(0)i} \right|
>4\sqrt{\log n}\|\mathbf{g}_{(i,r)}\|_{2}\right|\mathbf{g}_{(i,r)}\right\}\\
\leq&Pr\left\{\left.\left|\mathbf{g}_{(i,r)}'\mathcal{P}_{T}\mathbf{A}_{(i)}'\textrm{sgn}(\bar{\mathbf{s}}_{i})\right|
      >2\sqrt{\log{n}}\|\mathbf{g}_{(i,r)}\|_{2}\right|\mathbf{g}_{(i,r)}\right\}\\
&+Pr\left\{\left.\left|\mathbf{g}_{(i,r)}'\bar{\mathbf{v}}_i\right|\geq 2\sqrt{\log n}\|\mathbf{g}_{(i,r)}\|_2
    \right|\mathbf{g}_{(i,r)}\right\}
\leq4n^{-2}.
\end{align*}
Notice that because we bound the probability conditioned on
$\mathbf{g}_{(i,r)}$, the bound hold for any $j=1,\cdots,l$ and
any $r\in K_{ij}$. Now take a union bound over all $i=1,\cdots,L$,
\begin{align*}
&Pr\left\{\left.\bigcup_{i=1}^L\left\{\left| \frac{\sqrt{m}}{m_j}\mathbf{g}_{(i,r)}'\mathbf{q}_{(0)i} \right|
>4\sqrt{\log n}\|\mathbf{g}_{(i,r)}\|_{2}\right\}\right|\mathbf{g}_{(i,r)}\right\}\\
\leq&\sum_{i=1}^LPr\left\{\left.\left| \frac{\sqrt{m}}{m_j}\mathbf{g}_{(i,r)}'\mathbf{q}_{(0)i} \right|
>4\sqrt{\log n}\|\mathbf{g}_{(i,r)}\|_{2}\right|\mathbf{g}_{(i,r)}\right\}\\
\leq& L\cdot4n^{-2}\leq4n^{-1},
\end{align*}
where the last inequality follows from $k_TL\leq n$. Since the
right-hand side does not depend on $\mathbf{g}_{(i,r)}$ and the
inequality holds for any $j=1,\cdots,l$, any $r\in K_{ij}$, and
any $i=1,\cdots,L$, with probability at least $1-4n^{-1}$ it
follows
\begin{align}\label{certificate_4_conditional_bound}
\left| \frac{\sqrt{m}}{m_j}\mathbf{g}_{(i,r)}'\mathbf{q}_{(0)i} \right|
\leq4\sqrt{\log n}\|\mathbf{g}_{(i,r)}\|_{2}.
\end{align}

Next, we bound $\|\mathbf{g}_{(i,r)}\|_{2}$ using contractions
\eqref{Q_contraction_1}-\eqref{Q_contraction_2}. According to
Lemma \ref{batch_shrink_lemma}, with probability at least
$1-2n^{-1}$, \eqref{Q_contraction_1}-\eqref{Q_contraction_2} hold
simultaneously. Thus, with probability at least $1-2n^{-1}$, for
any $j\geq 3$, any $r\in K_{ij}$, and any $i=1,\cdots,L$, it holds
\begin{align*}
    \|\mathbf{g}_{(i,r)}\|_{2}
    \leq&\|\mathbf{a}_{(i)r}\|_2\left\|\left(\prod_{k=1}^{j-1}\mathcal{P}_{T}\left(\mathbf{I}
         -\tilde{\mathbf{A}}_{(1,k)}\right)\mathcal{P}_{T}\right)\right\|_{(2,2)}\\
    \leq&\frac{1}{\log n}\frac{1}{2^{j-1}}\sqrt{k_{T}\mu_{\max}}\leq\frac{1}{\log^{2}n}\sqrt{\frac{\alpha m}{\kappa_{\max}}},\nonumber
\end{align*}
given $k_{T}\leq \alpha\frac{m}{\mu_{max}\kappa_{max}\log^{2} n}$.
Thus, combining with \eqref{certificate_4_conditional_bound} gives
\begin{align*}
   \left|\frac{\sqrt{m}}{m_j}\mathbf{g}_{(i,r)}'\mathbf{q}_{(0)i} \right|
    &\leq \frac{m}{m_j}4\left(\log^{-\frac{3}{2}}{n}\right)\sqrt{\frac{\alpha}{\kappa_{\max}}}\\
    &\leq\frac{16}{\sqrt{9600}}\left(\log^{-\frac{1}{2}}{n}\right)\frac{1}{\sqrt{\kappa_{\max}}}\\
    &=\frac{2}{5\sqrt{6}}\frac{\lambda}{\sqrt{\kappa_{\max}}}\leq\frac{\lambda}{4},
\end{align*}
given $\alpha\leq\frac{1}{9600}$.

On the other hand, for any $j\leq 2$, any $r\in K_{ij}$, and any
$i=1,\cdots,L$, it holds
\begin{align*}
    \|\mathbf{g}_{(i,r)}\|_{2}
    \leq&\|\mathbf{a}_{(i)r}\|_2\leq\sqrt{k_{T}\mu_{\max}}\leq\frac{1}{\log n}\sqrt{\frac{\alpha m}{\kappa_{\max}}},\nonumber
\end{align*}
given $k_{T}\leq \alpha\frac{m}{\mu_{max}\kappa_{max}\log^{2} n}
$. Thus,  combining with \eqref{certificate_4_conditional_bound}
again gives
\begin{equation}
       \left|\frac{\sqrt{m}}{m_j}\mathbf{g}_{(i,r)}'\mathbf{q}_{(0)i} \right|
       \leq \frac{m}{m_i}4(\log^{-\frac{1}{2}}{n})\sqrt{\frac{\alpha}{\kappa_{max}}}
       \leq\frac{2}{5\sqrt{6}}\frac{\lambda}{\sqrt{\kappa_{\max}}}\leq\frac{\lambda}{4},
\end{equation}
given $\alpha\leq\frac{1}{9600}$. Hence, we finish the proof.
Notice that this bound requires \eqref{batch_inequalities} and
\eqref{certificate_4_conditional_bound} to hold simultaneously.
\end{IEEEproof}

\noindent \textcircled{5} \textbf{Estimation of the total success
probability.}

So far, we have proved that \textcircled{1}, \textcircled{2},
\textcircled{3}, \textcircled{4} hold with a high probability,
respectively. We want a success probability in recovering the true
signal, which not only requires \textcircled{1}, \textcircled{2},
\textcircled{3}, \textcircled{4} to hold simultaneously, but also
requires \eqref{BEQ_corollary}, \eqref{BEQ_corollary_plus},
Corollary \ref{BEQ_corollary_plusplus}, and Lemma
\ref{lemma_BEQ_2} to succeed. From the above proofs, we have
\begin{itemize}
  \item The bound \textcircled{1} is implied by \eqref{suff_condition_1} (holds with probability $1-4n^{-1}$).
  \item The bound \textcircled{2} is implied by \eqref{batch_inequalities} (holds with probability $1-2n^{-1}$) and \eqref{suff_condition_1}.
  \item The bound \textcircled{3} is implied by \eqref{batch_inequalities}, \eqref{suff_condition_1} and \eqref{appendix_certificate_3_1} (holds with probability $1-e^{\frac{1}{4}}n^{-1}$)
  \item The bound \textcircled{4} is implied by \eqref{batch_inequalities} and \eqref{certificate_4_conditional_bound} (holds with probability $1-4n^{-1}$).
\end{itemize}
Thus, we take a union bound to get
\[Pr\{\textcircled{1}\cup\textcircled{2}\cup\textcircled{3}\cup\textcircled{4}\}
\geq 1-4n^{-1}-2n^{-1}-e^{-\frac{1}{4}}n^{-1}-4n^{-1}
=1-\left(10+e^{\frac{1}{4}}\right)n^{-1}.\] On the other hand,
taking a union bound over \eqref{BEQ_corollary},
\eqref{BEQ_corollary_plus}, Corollary
\ref{BEQ_corollary_plusplus}, and Lemma \ref{lemma_BEQ_2} to find
that they hold simultaneously with probability at least
$1-\left(6+e^{\frac14}\right)n^{-2}$. Summarizing the above
results, we know that the success probability in recovering the
true signal and error matrices is at least
$1-(16+2e^{\frac14})n^{-1}$.

\bibliographystyle{unsrt}
\bibliography{bibliography}

\end{document}